\documentclass[a4paper,12pt]{article}

\usepackage{epsfig}
\usepackage{amsfonts,eucal}
\usepackage{amssymb,amsmath,amsthm}
\usepackage[dvipsnames]{xcolor}
\usepackage[T2A]{fontenc}
\usepackage[cp1251]{inputenc}
\newtheorem{remark}{Remark}[section]
\newtheorem{theorem}{Theorem}[section]
\newtheorem{proposition}{Proposition}[section]
\newtheorem{lemma}{Lemma}[section]

\newcommand{\Ga}{\Gamma}
\newcommand{\ga}{\gamma}
\newcommand{\BB}{{\mathbb B}}
\newcommand{\XX}{{\mathbb X}}
\newcommand{\R}{{\mathbb R}}
\newcommand{\N}{{\mathbb N}}
\newcommand{\X}{{\R^d}}
\newcommand{\La}{\Lambda}
\newcommand{\la}{\lambda}
\newcommand{\B}{\mathcal{B}}
\newcommand{\K}{\mathcal{K}_{\rm{adm}}}
\newcommand{\M}{\mathcal{M}}
\newcommand{\CB}{\mathcal B}

\begin{document}

\title{Invariant measures for spatial contact model in small dimensions}

\author{Yuri Kondratiev\thanks{Fakultat fur
Mathematik, Universitat Bielefeld, 33615 Bielefeld, Germany
(kondrat@math.uni-bielefeld.de).} \and Oleksandr Kutoviy\thanks{Fakultat fur
Mathematik, Universitat Bielefeld, 33615 Bielefeld, Germany
(kutoviy@math.uni-bielefeld.de).} \and Sergey Pirogov\thanks{
Institute for Information Transmission Problems, Moscow, Russia
(s.a.pirogov@bk.ru).} \and Elena Zhizhina\thanks{Institute for
Information Transmission Problems, Moscow, Russia (ejj@iitp.ru).} }

\date{}

\maketitle

\begin{abstract}
We study invariant measures of continuous contact model in small dimensional
spaces ($d =1,2$). We prove that this system has the one-parameter set of invariant measures
in the critical regime provided the dispersal kernel has a heavy tail. The convergence
to these invariant measures for a broad class of initial states is established.
\\

\end{abstract}

\section{Introduction}

In the present paper we are dealing with a continuous analog of the well-known lattice contact process \cite{Lig1985}. The continuous contact process is a particular case of the general birth-and-death processes in the continuum. The existence problem for the continuous version of the lattice contact model in terms of the corresponding spatial Markov process was thoroughly analyzed in \cite{KS}. In the spatial plant ecology such process describes a Markov evolution of a plant population with an independent seed production by each parent plant (accordingly to a dispersal probability density $0\leq a\in\mathrm{L}^{1}(\R^{d})$) and independent exponentially distributed random life time with parameter 1 for each of them (global mortality rate).
One of the main features of this model is the clustering of the system, i.e. particles are grouped into large clouds of high density, which are located at large distances from each other.  It is worth noting that the appearance of a limiting invariant state is only possible in the so-called critical regime (i.e., there is a certain balance between birth and death). As it was recently shown in \cite{KKP}, in the case of the critical regime and $d\geq 3$, there exists continuum of invariant measures parametrized by the density values. These invariant measures are described by a simple recurrent relation between their correlation functions and create a concrete (and up to our knowledge, completely new) class of random point fields. For all other regimes, the density of the system tends either to $\infty$ or to $0$ as time grows. The existence of a stationary regime in the marked contact model with a compact spin space was proved in \cite{KPZh}.

The present paper concerns the asymptotic behavior and invariant measures for continuous contact model in the cases $d=1,2$.  For such low dimensional situation, it is proved that the continuous contact model has invariant measure if the tail of the dispersal kernel $a$ is "heavy" enough. In this case, the critical contact process starting with an admissible initial state converges to the equilibrium measure uniquely defined by the density of the initial state. Note, that the restriction to the class of dispersal kernels with "heavy" tails for $d=1,2$ is hypothetically related to the instability of the spectrum of Schr\"odinger operator at low dimensions.
In 1929, R. Peierls made a crucial discovery that at low dimensions ($d=1,2$), contrary to the three dimensional situation, an arbitrary small potential well leads to the emerging of the bound state. For the non local Schr\"odinger operators corresponding to the contact model such behavior was
also established in \cite[Theorem 9]{KMPZ}. It means that even small local deviation downwards from the critical mortality level leads to the exponential growth of the population whenever  $d=1,2$ and dispersal kernels have sufficiently light tails. Heavy tails of dispersal kernels appear to make the critical regime more stable contrary to light tails.

For the critical contact model with light tails in low dimensions the typical configuration has the following
structure as $t\to\infty$: there is a collection of “over-populated cities” separated by large empty planes. The pair correlation function grows unlimitedly as $t \to \infty$ and so the distribution of the corresponding random field has no limit, see Remark \ref{Remark3} in Section 5.  Thus the light tail (or the existence of the second
moment) of the dispersal kernel in small dimensions
results to the phenomenon of "clustering". Similar
results have been obtained for a certain type of critical branching random walks, see e.g. \cite{BS}, \cite{GW}.

In this work we formulate conditions on dispersal kernels which ensure the existence
of correlations functions of the corresponding stationary measures of the model in the critical regime in low dimensions.
We also specify relations between solutions to the corresponding Cauchy
problem and these stationary regimes.

\section{Main results}
\subsection{The model}

%\subsection{Homogeneous mortality rates}

Let ${\B}({\X})$ be the family of all Borel sets in ${\X}$, $d\geq
1$ and let  ${\B}_{\mathrm{b}} ({\X})$ denote the system of all bounded
sets from ${\B}({\X})$.

The continuous contact model is regarded as a spatial Markov process which is a
particular case of the general birth-and-death process in the continuum, see e.g.
\cite{KKP, KS}. The phase space of such processes is the space of locally finite configurations in $\R^{d}$. Namely,
\begin{equation} \label{confspace}
\Ga =\Ga\bigl(\X\bigr) :=\Bigl\{ \ga \subset \X \Bigm| |\ga\cap\La
|<\infty, \ \mathrm{for \ all } \ \La \in {\B}_{\mathrm{b}}
(\X)\Bigr\}.
\end{equation}
Here $|\cdot|$ denotes the number of elements of a~set. We can identify each
$\ga\in\Ga$ with the non-negative Radon measure $\sum_{x\in
\gamma }\delta_x\in \M(\X)$, where $\delta_x$ is
the Dirac measure with unit mass at $x$, $\sum_{x\in\emptyset}\delta_x$ is, by definition, the zero measure, and $\M(\X)$ denotes the space of all
non-negative Radon measures on $\B(\X)$. This identification allows us to endow
$\Ga$ with the topology induced by the vague topology on
$\M(\X)$, i.e. the weakest topology on $\Ga$
with respect to which all the mappings
\begin{equation}\label{gentop}
    \Ga\ni\ga\mapsto \sum_{x\in\ga} f(x)\in{\R}
\end{equation}
are continuous for any $f\in C_0(\X)$ that is the set of all continuous functions on $\X$ with compact supports.  It is worth noting that
the vague topology can be metrizable in such a way that $\Ga$ becomes a~Polish space (see e.g. \cite{KK2006} and references therein). The topological space $\Gamma(X)$ for any $X\in \CB(\X)$ can be defined in a similar way. The Borel $\sigma$-algebra on $\Gamma(X)$ is denoted by $\CB(\Gamma(X))$.

The spatial contact model is given by a heuristic
generator defined on a proper class of functions  $F:
\Gamma \to \R$ as follows:
\begin{align}
 \begin{aligned}\label{generator}
(L F)(\gamma) &= \sum_{ x \in \gamma}\left[F(\gamma \backslash \{x\}) - F(
\gamma)\right] \\
&+   \int\limits_{\R^{d}} \sum_{x \in \gamma}  a(x-y) (F(\gamma \,\cup
\{y\} ) - F(\gamma)) dy,
 \end{aligned}
 \end{align}
where $dx $ is the Lebesgue measure on $\R^{d}$. In the sequel, for simplicity of notations, we just write $x$
instead of $\{x\}$. The first term in \eqref{generator} corresponds to the death of the particles. Namely, each $x$ of the configuration $\gamma\in\Ga$ dies with the death rate $1$. The second term of \eqref{generator} describes the birth of a new particle at the point $y$ with the
 birth rate density $b(y,\gamma):=\sum_{x \in \gamma} a(x-y)$.

We assume that  $a$ is a non-negative function on $\R^{d}$ satisfying the following conditions:
\begin{enumerate}
\item {\it Critical regime} condition:
\begin{equation}\label{2a}
\| a \|_{L^1} \ = \  \int\limits_{\R^{d}} a(u) du \ = \ 1;
\end{equation}
\item {\it Regularity} condition:
\begin{equation}\label{2b}
 \hat a (p) \ := \ \int\limits_{\R^{d}} e^{-i(p,u)} a(u) du \in
L^1(\R^{d}),
\end{equation}
where the symbol $(\cdot\,,\cdot)$ stands for scalar product in $\X$;
\item {\it Heavy tail} condition:
\begin{equation}\label{2c.1}
a(x) \ \sim \ \frac{1}{|x|^{\alpha+1}} \quad \mbox{as } \; |x| \to \infty, \; 0<\alpha<1, \; \mbox{ in the case } \; d=1;
\end{equation}
\begin{equation}\label{2c.2}
a(x) \ \sim \ \frac{1}{|x|^{\alpha+2}} \quad \mbox{as } \; |x| \to \infty, \; 0<\alpha<2, \; \mbox{ in the case } \; d=2.
\end{equation}
\end{enumerate}

\begin{remark}

(a) Conditions $1$ and $2$ imply that $a$ and $\hat a$ are bounded continuous functions vanishing at infinity. Moreover, the non-negativity of $a$  yields
\begin{equation}\label{less1}
|\hat a (p)|<1,\quad \textrm{for all} \quad p \neq 0.
\end{equation}

(b) Condition $3$ is crucial to ensure the convergence of the integral (\ref{L1}).

(c) In general, the function $a$ is not even.

\end{remark}

\subsection{Basic facts and notations}

For the technical purposes related to the approach needed to derive time
evolution equations for correlation functions (see e.g. \cite{KM2008}), we consider the space of finite configurations whose natural topology is different from the vague one considered on $\Gamma$. The  space of finite configuration  is defined  by
%Let ${\cal B}(\R^{d})$ be the family of all Borel sets in $\R^{d}$, and
%${\cal B}_b (\R^{d}) \subset {\cal B}(\R^{d})$ denotes the
%family of all bounded sets from ${\cal B}(\R^{d})$.
$$
\Gamma_0   =  \Gamma_0 (\R^{d})  =  \bigsqcup_{n \in \N \cup \{0 \} }
\ \Gamma_0^{(n)},
$$
where
$$
\Gamma_0^{(n)}  =  \Gamma_{0,\X}^{(n)}  =  \{ \eta \subset \R^{d} : \ |\eta| = n \}
$$
is the space of $n$-point configurations. The space $\Gamma_0$ is equipped with the topology of the disjoint union. We denote the corresponding  Borel $\sigma$-algebra to this topology by $\mathcal{B}(\Gamma_0)$.
The space of $n$-point configurations in $Y\in\CB(\X)$, denoted by $\Gamma_{0,Y}^{(n)}$, can be defined analogously to $\Gamma_{0,\X}^{(n)}$. For the space of finite configurations in $Y\in\CB(\X)$ we will use the symbol $\Gamma_{0}(Y)$.
It is worth noting that $\Gamma_{0}$ is a subset of $\Gamma$.

Next, we describe the classes of functions and measures on $\Gamma_{0}$ and $\Gamma$ which will be used in the sequel. A set $M\in\mathcal{B}(\Gamma_0)$ is called bounded if there exists $\Lambda\in\CB_b(\X)$ and $N\in\N$ such that $M\subset\bigsqcup_{n=1}^{N}\Gamma_{0,\Lambda}^{(n)}$. The class $B_{bs}(\Gamma_0)$ stands for the set of all bounded measurable functions on $\Gamma_0$ which have bounded support, i.e., $G\in B_{bs}(\Gamma_0)$ if $G|_{\Gamma_0\setminus M}=0$ for some bounded $M\in \mathcal{B}(\Gamma_0)$. The class
${\cal F}_{cyl}(\Gamma)$ denotes the set of cylindrical functions on $\Gamma$, i.e., the set of all measurable functions $F$ on $(\Gamma,\CB(\Gamma))$, which are measurable with respect to $\CB(\Gamma({\Lambda}))$ for some $\Lambda\in \CB_b(\X)$. Any $F\in{\cal
F}_{cyl}(\Gamma)$ determines a set $\Lambda\in\CB_{b}(\X)$ such that $F(\gamma)=F(\gamma_\Lambda)$ for all $\gamma\in\Gamma$.

%Let $\sigma$ be a non-atomic Radon measure on $\bigl(\X,\CB(\X)\bigr)$, e.g.  the Lebesgue measure.
The analog of Lebesgue measure on $(\Gamma_0,\CB(\Gamma_0))$ is the {\it Lebesgue-Poisson} measure. It is defined on $(\Gamma_0, \CB(\Gamma_{0}))$ by \[\lambda_z:=\delta_{{\emptyset}}+\sum_{n=1}^\infty \frac{z^n}{n!}m^{(n)},\quad z>0,\]
where $m^{(n)}$ is the projection of the product Lebesgue measure $(dx)^{n}$ considered on $(\X)^{n}$ to $(\Gamma_0^{(n)}, \CB(\Gamma_{0}^{(n)}))$. Throughout the paper, we take the parameter $z$ to be 1 using the notation $\lambda:=\lambda_1$ for this case.

%The space $ \Gamma_0 (\R^{d}) $ is equipped by the Lebesgue-Poisson measure $\exp (dx) = 1 + dx + \frac{dx \otimes dx}{2!}+ \ldots $.
%We denote the set of bounded measurable functions with bounded
%support by $B_{bs}(\Gamma_0)$, and the set of cylinder functions on
%$\Gamma$ by ${\cal F}_{cyl}(\Gamma)$. Each $F \in {\cal
%F}_{cyl}(\Gamma)$ is characterized by the following relation:
%$F(\gamma) = F(\gamma_{\Lambda})$ for some $\Lambda \in {\cal B}_b
%(\R^{d})$.
%The notation $B_{bs}(\Gamma_0)$ is used for
%the set of bounded measurable functions with bounded support,
%i.e.  $G \in B_{bs}(\Gamma_0)$, if $G$ is a bounded measurable function on $\Gamma_0$, and there exists
%$\Lambda \in {\cal B}_b(\R^{d})$ and $N \in \mathbb{N}$ such that
%\begin{equation}\label{F4}
%G|_{\Gamma_0 \backslash \ \bigsqcup_{n=0}^N \ \Gamma_\Lambda^{(n)}} = 0,
%\end{equation}
%where  $\Gamma_\Lambda^{(n)} \ = \ \{ \eta \subset \Lambda : \ |\eta| = n\}$ is the space of $n$-point configurations from $\Lambda$.

Next, we introduce the  mapping between $B_{bs}(\Gamma_0)$ and ${\cal F}_{cyl}(\Gamma)$, which turns out to be crucial for studying functions and measures on $\Gamma_{0}$ and $\Gamma$. For any function $G\in B_{bs}(\Gamma_0)$, we define the $K$-transform of $G$ as
\begin{equation}
KG(\gamma):=\sum_{\eta\Subset\gamma}G(\eta),\quad \gamma\in \Gamma.\label{K-transform}
\end{equation}
Here, with the notation $\eta\Subset \gamma$ we have the summation over all finite subconfigurations $\eta
\in \Gamma_0$ of the infinite configuration $\gamma \in \Gamma$. One has to emphasize that this mapping is linear, positivity preserving and injective (see, e.g. \cite{KK2002}).

%In the same way as in \cite{KKP} we conclude that
%the operator $\hat L = K^{-1}LK$ (the
%image of $L$ under the K-transform) on functions   $G \in
%B_{bs}(\Gamma_0)$ has the following form:
%\begin{equation}\label{prop1}
%(\hat L G)(\eta) \ = \ - |\eta| G(\eta)  \  + \  \int\limits_{\R^{d}}
%\sum_{x \in \eta} a(x-y) G((\eta \backslash x) \cup y) dy \  +
%\end{equation}
%$$
%\int\limits_{\R^{d}} \sum_{x \in \eta} a(x-y) G(\eta \cup y) dy.
%$$
%The derivation of the formula (\ref{prop1}) is the same as in\cite{KKP}.

%Remind that $|\eta| = \sum_{x \in \eta} m(x)$ under the choice of mortality
%rates $m(x) \equiv 1$ in formula (\ref{generator

Denote by ${\cal M}^1_{fm}(\Gamma)$ the set of all probability
measures $\mu$ which have finite local moments of all orders, i.e.
$$
\int_{\Gamma} |\gamma_{\Lambda}|^n \ \mu (d \gamma) \ < \ \infty
$$
for  all $\Lambda \in {\cal B}_b(\R^{d})$ and $n \in N$.  The {\em Poisson measure} $\pi$ on $\Gamma$ with intensity measure $dx$ is an example of such measure. It  is defined on
$(\Gamma, \CB(\Gamma))$ by
  \begin{equation}\label{Poisson}
    \pi \bigl(\Gamma_{0,\Lambda}^{(n)}\bigr)=\frac{\bigl(\rm{vol}(\Lambda)\bigr)^n}{n!}\exp\bigl\{ -\rm{vol}(\Lambda)\bigr\}, \qquad \Lambda\in\CB_{b}(\X), \quad n\in\N\cup\{0\},
  \end{equation}
  where $\rm{vol}(\Lambda)$ is the Lebesgue mass of $\La$.  Note that \eqref{Poisson} determine the Poisson measure on $\B(\Ga)$ uniquely.

If a measure
$\mu \in {\cal M}^1_{fm}(\Gamma)$ is locally absolutely continuous
with respect to the Poisson measure, then there exists the {\it system of correlation
functions} $k_\mu:\Gamma_{0}\to\R_+$ of the measure $\mu$, see e.g. \cite[Chapter 4]{R}.
For all $G\in B_{bs}(\Gamma_{0})$, it satisfies
\begin{equation}\label{F6}
\int_{\Gamma} (K G)(\gamma) \mu (d \gamma) = \int_{\Gamma_{0}} G(\eta)k_{\mu}(\eta) \lambda (d \eta).
%=:\langle!\langle G, k_\mu \rangle!\rangle.
\end{equation}
\begin{remark}
%It is worth pointing out that the structure of $\Ga_{0}$ implies that any function $f$ on $\Ga_{0}$ can be regarded as a collection of functions $(f^{(n)})_{n\in\N\cup\{0\}}$, where $k^{(n)}$ is a function on $(\X)^{n}$. Hence, a
Any $k_{\mu}:\Gamma_{0}\to\R_+$ can be written as a sequence of functions $k^{(n)}\colon(\X)^n\to\R_{+}$ defined by
 \[
  k^{(n)}(x_1,\dotsc,x_n)=\begin{cases}
                                                           k_{\mu}(\{x_1,\dotsc,x_n\}),&\text{ if }(x_1,\dotsc,x_n)\in\widetilde{\left(\mathbb{R}^{d}\right)^n},\\
                                                           0,&\text{ otherwise},
                                                          \end{cases}
 \]
 \[\widetilde{(\mathbb{R}^{d})^{n}}:=\Big\{(x_1,\cdots,x_n)\in (\R^{d})^{n}~|~ x_k\neq x_l,\quad \mbox{if}\quad k\neq l\Big\}.\]
 The function $k^{(n)}$ is called the $n$-point correlation function of the measure $\mu$.
\end{remark}
It will cause no confusion if we use function on $\Ga_{0}$ and the collection of symmetric functions on $\widetilde{\left(\mathbb{R}^{d}\right)^n}$ in the sense of the previous remark. We define ${\cal M}^1_{\rm{corr}}(\Gamma)$ to be the subclass of ${\cal M}^1_{fm}(\Gamma)$ consisting of those probability measures on $\Ga$ for which the corresponding correlation functions exist.

\subsection{Time evolution of correlation functions}

Next we follow the general scheme to study the forward Kolmogorov equation for the continuous contact model proposed in \cite{KKP}. For the convenience of the reader we repeat below some details of this approach, thus making our exposition self-contained.

The existence problem for a  Markov process with a priori given form of a generator $L$ is a challenging problem in general. On the other hand, the evolution of an initial  distribution in the course of a stochastic dynamics is an important object and it deserves a special attention. The existence of such evolution in our case may be realized through
the {\em forward Kolmogorov} (or {\em Fokker--Planck}) equation with the evolution operator $L$ for probability measures (states) on the configuration space $\Gamma$, i.e.
\begin{equation}\label{FPE-init}
  \frac{d}{d t} \mu_t(F) = \mu_t(LF),
  \quad t>0, \quad \mu_t\bigr|_{t=0}=\mu_0,
\end{equation}
where $$\mu(F):=\int_\Ga F(\ga)\,d\mu(\ga)$$ and  $F$ is an arbitrary function from an appropriate set ${\cal F}$ for which both sides of \eqref{FPE-init} make sense. For the purposes of the proposed approach we assume that $\cal F$ includes $K\bigl(B_{bs}(\Gamma_0))$.
%, e.g. $$F\in K\bigl(B_{bs}(\Gamma_0)\bigr)\subset{\cal F}_{cyl}(\Gamma).$$
The mere existence of the solution to \eqref{FPE-init} does not necessarily implies the existence of the corresponding correlation function at each moment of time. Having in mind applications, it is very important to construct the evolution of initial correlation function as it provides the most of statistical characteristics of the studied process. Hence, we suppose now that a solution $\mu_t$ to \eqref{FPE-init}
exists and $\mu_t\in {\cal M}^1_{\rm{corr}}(\Gamma)$ for any  $t>0$
%and remains locally absolutely continuous with respect to
%the Poisson measure $\pi$ for all $t>0$
provided $\mu_0\in {\cal M}^1_{corr}(\Gamma)$.
Then, we consider the corresponding correlation function
$k_t:=k_{\mu_t}$ for any $t\geq0$.

Assume one can calculate
$K^{-1}LF$ for $F\in{\cal F}$. In this case, we are able to rewrite
\eqref{FPE-init} as follows
\begin{equation}\label{ssd0}
  \frac{d}{d t} \langle K^{-1}F, k_t\rangle
  = \langle K^{-1}LF, k_t\rangle,\quad t>0, \quad
  k_t\bigr|_{t=0}=k_0,
\end{equation}
for all $F\in {\cal F}$ for which both sides of \eqref{FPE-init} make sense. Here the pairing
between functions on $\Ga_0$ is given by
\begin{equation}
\left\langle G,\,k\right\rangle
:=\int_{\Ga _{0}}G(\eta) k(\eta) \,d\la(\eta). \label{duality}
\end{equation}
Let us recall that by the definition of the Lebesgue-Poisson measure we have
\begin{equation*}\label{duality-intro}
  \langle G,k \rangle=
  \sum_{n=0}^\infty \frac{1}{n!} \int_{(\X)^n}
  G^{(n)}(x_1,\ldots,x_n)
  k^{(n)}(x_1,\ldots,x_n)\,dx_1\ldots dx_n,
\end{equation*}
Next, if we substitute $F=KG$, $G\in
B_{bs}(\Gamma_0)$ into \eqref{ssd0} (recall that $K\bigl(B_{bs}(\Gamma_0))\subset \cal F$), we derive
\begin{equation}\label{ssd}
  \frac{d}{d t} \langle G, k_t\rangle
  = \langle \widehat{L}G, k_t\rangle, \quad t>0, \quad
  k_t\bigr|_{t=0}=k_0,
\end{equation}
for all $G\in B_{bs}(\Gamma_0)$. Here we suppose that the operator
\begin{equation*}\label{Lhat}
    (\widehat{L}G)(\eta ) :=( K^{-1}LKG)(\eta),\quad \eta\in\Ga_0
\end{equation*}
is defined at least point-wisely for all $G\in B_{bs}(\Gamma_0)$.
As a result, we will be interested now in a solution to the equation
\begin{equation}\label{QE}
  \frac{\partial k_t}{\partial t}
  = \widehat{L}^* k_t, \quad t>0, \quad
  k_t\bigr|_{t=0}=k_0,
\end{equation}
where $\widehat{L}^*$ is dual operator to $\widehat{L}$ with respect to the
duality \eqref{duality}, i.e.,
\begin{equation}
\int_{\Ga _{0}}(\widehat{L}G)(\eta) k(\eta) \,d\la(\eta)
=\int_{\Ga _{0}}G(\eta) (\widehat{L}^*k)(\eta) \,d\la(\eta). \label{dualoper}
\end{equation}
%The procedure to obtain the operator $\widehat{L}$ for a given $L$ is fully combinatorial meanwhile to obtain the expression for the operator $\widehat{L}^*$ we need
%an analog of integration by parts formula. For a difference operator $L$ considered in \eqref{genGen} this discrete integration by parts rule is presented in the well-known lemma (see e.g. \cite{KMZ2004}):
%\begin{lemma}
%\label{Minlos} For any measurable function $H:\Ga_0\times\Ga_0\times
%\Ga_0\rightarrow{\R}$
%\begin{equation} \label{minlosid}
%\int_{\Ga _{0}}\sum_{\xi \subset \eta }H\left( \xi ,\eta \setminus
%\xi ,\eta \right) d\la \left( \eta \right) =\int_{\Ga _{0}}\int_{\Ga
%_{0}}H\left( \xi ,\eta ,\eta \cup \xi \right) d\la \left( \xi
%\right) d\la \left( \eta \right)
%\end{equation}
%if at least one side of the equality is finite for $|H|$.
%\end{lemma}
%In particular, if $H(\xi,\cdot,\cdot)\equiv0$, $|\xi|\neq1$ we obtain an analog of \eqref{Mecke}, namely,
% \begin{equation}\label{laMecke}
% \int_{\Ga_0} \sum_{x\in\eta} h(x,\eta\setminus x,\eta)d\la(\eta)
% =\int_{\Ga_0} \int_\X h(x,\eta,\eta\cup x) dx d\la(\eta),
%\end{equation}
%for any measurable function $h:\X\times\Ga_0\times
%\Ga_0\to\R$ such that both sides make sense.

For the derivation of the formula for $\hat L $ as well as the formula for $\widehat{L}^*$ in the case of $L$ given by \eqref{generator} we refer the reader  to \cite{KKP}. According to the latter reference, the operator $\hat L $
%Now, in the same way as in \cite{KKP}, we conclude that
%the operator $\hat L $ corresponding to
%the continuous contact model with
%$L$ given by \eqref{generator}
has the following form on functions  $G \in
B_{bs}(\Gamma_0)$:
\begin{align}
 \begin{aligned}\label{prop1}
(\hat L G)(\eta)  = & - |\eta| G(\eta)  \  + \  \int\limits_{\R^{d}}
\sum_{x \in \eta} a(x-y) G((\eta \backslash x) \cup y) dy \\
& + \int_{\R^{d}} \sum_{x \in \eta} a(x-y) G(\eta \cup y) dy.
\end{aligned}
\end{align}
%For the derivation of the formula (\ref{prop1}) as well as the formula for $\widehat{L}^*$ in the case of $L$ given by \eqref{generator} we refer the reader  to \cite{KKP}. According to the latter reference,
%Remind that $|\eta| = \sum_{x \in \eta} m(x)$ under the choice of mortality
%rates $m(x) \equiv 1$ in formula (\ref{generator}).
%Let $\{ \mu_t \}_{t \ge 0} \subset {\cal M}_{fm}^1
%(\Gamma)$ be the evolution of states described by the forward Kolmogorov equation with
%the adjoint operator $L^{\ast}$.
%Then the generator of the evolution of the
%corresponding system of correlation functions is defined as
%The formal generator of evolution of correlation functions is defined as
%\begin{equation}\label{F7}
%\langle \hat L G, k \rangle \ = \  \langle G, \hat L^{\ast} k
%\rangle, \quad  G \in B_{bs}(\Gamma_0),
%\end{equation}
%where the operator $\hat L = K^{-1} L K$ is defined by (\ref{prop1}).
The evolution equations for the system of $n$-point correlation functions corresponding to the continuous contact model have the following recurrent forms:
\begin{equation}\label{59}
\frac{\partial k_{t}^{(n)}}{\partial t} \ = \ \hat L_n^{\ast} k_{t}^{(n)} \
+ \ f_{t}^{(n)}, \quad n\ge 1; \qquad k_{t}^{(0)} \equiv 1.
\end{equation}
Here $f_{t}^{(n)}$ are functions on $(\R^{d})^{n}$ defined for $n \ge 2$ by
\begin{equation}\label{f}
f_{t}^{(n)}(x_1, \ldots, x_n) \ = \ \sum_{i=1}^n  k_{t}^{(n-1)}(x_1,
\ldots,\check{x_i}, \ldots, x_n) \sum_{j\neq i}^n a(x_i - x_j),
\end{equation}
and  $f_{t}^{(1)} \equiv 0$. The form of the operator $\hat L^{\ast}_n, \; n \ge 1, $ is
given by
\begin{align}
 \begin{aligned}\label{korf}
 \hat L^{\ast}_n k^{(n)}(x_1, &\ldots, x_n)  =  - n \,k^{(n)}(x_1,
\ldots, x_n) \\
& + \sum_{i=1}^n \int\limits_{\R^{d}} a(x_i- y) k^{(n)}(x_1, \ldots, x_{i-1}, y,
x_{i+1}, \ldots, x_n) dy.
\end{aligned}
\end{align}

It is worth pointing out that the proper choice of a countably normed space for the Cauchy problem \eqref{QE} containing the system of correlation functions at any $t \ge 0$ is the important step of our constructions. The structure of this space is similar to the Fock space in quantum mechanics, but instead of the $L_2$-norm we use here the collection of sup-norms.
Note, that the choice of integrable correlation functions would mean that our stochastic dynamics evolves through finite configurations only. So we can not use $L_2$-norm (as well as $L_1$-norm) for our system.

%In the present paper we are dealing with infinite system of particles evolving according to the contact process. Therefore, we have to %consider the countably normed space for \eqref{QE} (equivalently, for the Cauchy problem \eqref{59}) consistent with such situation. The %corresponding countably normed space for \eqref{59} is introduced below.

Let $\BB((\R^{d})^n)$ be the Banach space of all measurable real-valued bounded functions on $(\R^{d})^n$ with the $\sup$-norm. We denote by
$\BB_{\rm{inv}}((\R^{d})^n)$ the subset of $\BB((\R^{d})^n)$, which are additionally translation invariant, i.e., a function $\varphi\in\BB_{\rm{inv}}((\R^{d})^n)$ if $\varphi\in \BB((\R^{d})^n)$ and for any $(w_1, \ldots, w_n)\in(\X)^{n}$
$$
\varphi (w_1+u, \ldots, w_n+u) = \varphi (w_1, \ldots, w_n), \quad \forall u \in \R^{d}.
$$
It is easily seen that $\BB_{\rm{inv}}((\R^{d})^n)$ is a closed subset in $\BB((\R^{d})^n)$ and, hence, it is the Banach space with  respect to the $\sup$-norm. By abuse of notation we continue to write $\XX_{n}$ for $\BB_{\rm{inv}}((\R^{d})^n)$, $n\geq 1$. The collection of sup-norms for $k^{(n)}$ defines the structure of countably normed space for the systems of correlation functions.

\begin{remark}
Consider the operator $\hat L_n^{\ast}$ as an operator on the Banach space $\XX_{n}$ for any $n\geq 1$. It is a simple matter to check that it is bounded linear operator in $\XX_{n}$ and in  $\BB((\R^{d})^n)$. The arguments similar to those in \cite[p.~436]{EN} show that the %classical
solution to the Cauchy problem \eqref{59} in  $\XX_{n}$ with arbitrary initial values $k_{0}^{(n)}\in\XX_{n}$ exists and is unique provided $f_{t}^{(n)}$ is constructed recurrently via the
%classical
solution to the same Cauchy problem \eqref{59} for $n-1$.
%It is given by the variation of parameters formula.
\end{remark}
%Therefore, the right-hand side of the Cauchy problem \eqref{59} is well-defined as an element of $\XX_{n}$. The choice of $\XX_{n}$ for \eqref{59} requires, in fact, an admissible class of initial states for which the statistical dynamics can be constructed.

To study the ergodic properties of the
%classical
solution to the system \eqref{59} in Banach spaces $(\XX_{n})_{n\geq 1}$  we assume that initial data belong to the following admissible class of functions. We denote this class as $\K$. It is defined by $$\K=\cup_{\varrho>0}\K(\varrho),$$ where
$\K(\varrho)$ consists of functions $k:\Gamma_{0}\to \R_{+}$ such that
%$k^{(0)}=1$ and
%there exists $\varrho\in\R_{+}$ such that
\begin{equation}\label{k0-1}
k^{(0)}\equiv1,\quad k^{(1)}  \equiv \varrho,
\end{equation}
\begin{equation}\label{k0}
k^{(n)} = \varrho^n + r^{(n)},  \quad n\ge 2.
\end{equation}
Here, $r^{(n)}\in \XX_n$ is a symmetric function satisfying for all $(x_1, \ldots, x_n)\in(\X)^{n}$
\begin{equation}\label{r0}
r^{(n)}(x_1, \ldots, x_n) \le D C^{n-1} ((n-1)!)^2 \ \sum_{\substack{i,\,j=1\\i \neq j}}^n r^{(n)}_{i j} ( x_i -  x_j),
\end{equation}
where  $r^{(n)}_{ij}: \X\to\R_{+}$,  $1\leq i,j\leq n$, $i\neq j$ are some functions such that
\begin{equation}\label{rij}
r^{(n)}_{i j} \in L^1(\R^{d}), \quad \hat r^{(n)}_{i j} \in L^1(\R^{d}),
\end{equation}
and $C$, $D$ are some positive constants.

In particular, it follows from the above conditions that both $r^{(n)}_{i j}$ and its Fourier transform $\hat r^{(n)}_{i j}$ are bounded continuous functions vanishing at infinity. As a result,
\begin{equation}\label{60}
|k^{(n)}(x_1, \ldots, x_n) - \varrho^n|   =    |r^{(n)}(x_1, \ldots, x_n) | \to  0,
\end{equation}
whenever  $|x_i - x_j| \to \infty$ for all $i \neq j$.
Note that the similar estimates for correlation functions were obtained for low density gases in \cite{DS}.
%The last relation (\ref{60}) means that the initial measure is an
%asymptotically Poisson measure with the intensity $\varrho$.
%\begin{definition}
%The solution to the Cauchy problem \eqref{59} in the Banach spaces $(\XX_{n})_{n\geq 1}$ is a function $k_{t}:\Ga_{0}\to\R$ which is continuously differentiable with respect to
%\end{definition}

According to the above mentioned scheme the invariant measures of the contact process belonging to the class
${\cal M}^1_{\rm{corr}}(\Gamma)$ are described in terms of the corresponding correlation functions
$\{k^{(n)}\}_{n\geq 0}$ as positive solutions to the following system:
\begin{equation}\label{Last}
\hat L^{\ast}_n k^{(n)} + f^{(n)}=0, \quad n \ge 1, \quad
k^{(0)}\equiv 1,
\end{equation}
where $\hat L_n^{\ast}, \, f^{(n)}$ are defined as in \eqref{f}-\eqref{korf}.

In the sequel, we say that $k:\Ga_{0}\to \R$ solves the system \eqref{Last} in the Banach spaces $(\XX_{n})_{n\geq 1}$ if the corresponding $k^{(n)}\in\XX_{n}$, $n\geq 1$ and $\{k^{(n)}\}_{n\geq 0}$ solves \eqref{Last}.

We prove the existence of a solution to the system (\ref{Last}) in the Banach spaces $(\XX_{n})_{n\geq 1}$, which
has a specified asymptotic whenever $|x_i - x_j|\to\infty$
for all $i \neq j$. Moreover,  we show that the
solutions to the Cauchy problem \eqref{59} with initial correlation functions belonging to $\K$ converge to the
solutions of (\ref{Last}) as time tends to infinity. These results are stated in the following theorem.

\begin{theorem}\label{mainth} {\it Let $d=1,2$. Assume that the dispersal kernel of the
contact model satisfies conditions \eqref{2a}-\eqref{2c.2}. Then the following assertions hold.

{(i)} For any positive constant $\varrho >0$ there exists a
unique probability measure $\mu^{\varrho}$ on $\Ga$ such that its
correlation function $k_{\varrho}: \Ga_{0}\to\R_{+}$ solves \eqref{Last}  in the Banach spaces $(\XX_{n})_{n\geq 1}$,  the corresponding system $\{k_{\varrho}^{(n)}\}_{n\geq 1}$ satisfies $k_\varrho^{(1)}\equiv \varrho$ and
\begin{equation}\label{Th1}
| k^{(n)}_\varrho (x_1, \ldots, x_n) \ - \ \varrho^n | \ \to \ 0, \ \text{whenever}  \ |x_i - x_j| \to \infty \ \text{for all} \  i \neq j.
\end{equation}
Moreover, there exist positive constants $C,\,D$ such that
\begin{equation}\label{estimate}
k^{(n)}_\varrho (x_1, \ldots, x_n) \ \le   D  C^n (n!)^2
\quad  \text{for all} \quad (x_1, \ldots, x_n)\in(\X)^{n}.
\end{equation}

{(ii)} Let $\{k_{t}^{(n)}\}_{n\geq 1}$ be the
%classical
solution to (\ref{59}) with initial value from $\K$ in the Banach spaces $(\XX_{n})_{n\geq 1}$. Then, there exists $\rho>0$ such that
%$k_{t}^{(n)}$ converges to the solution
%$k_\varrho^{(n)}$ (\ref{Th1}) of the system (\ref{Last}) of
%stationary (time-independent) equations as  $t \to \infty$:
\begin{equation}\label{Th1-2}
\| k_{t}^{(n)} \ - \ k_\varrho^{(n)} \|_{\XX_n} \ \to \ 0, \quad t \to \infty, \quad \forall n\geq 1.
\end{equation}
}
\end{theorem}

\section{The proof of Theorem \ref{mainth} (i). Stationary problem.}

In this section we prove the first part of Theorem \ref{mainth} using the
induction in $n\in\N$.
For $n=1$ in (\ref{Last}) we have
\begin{equation}\label{8}
-k^{(1)}(x) + \int\limits_{\R^{d}} a(x-y) k^{(1)}(y) dy = 0.
\end{equation}
It follows immediately that $k^{(1)} \equiv \varrho $ is an element of $\XX_{1}$ and it solves \eqref{8}.
%
%
%
%Since we are looking for a translation invariant solution, the condition (\ref{2a}) implies
%$$
%k^{(1)}(x) \ = \ \varrho.
%$$
We notice that $\varrho$ can be interpreted as the spatial density of particles.
To solve the equation \eqref{Last} for the case $n=2$ in $\XX_{2}$ we need the following lemma.

\bigskip

\begin{lemma}\label{lem1}
Conditions (\ref{2b}) - (\ref{2c.2}) imply
\begin{equation}\label{L1}
\int\limits_{\mathbb{R}^d} \frac{|\hat a(p)| \, dp}{2 - \hat a(p) - \hat a(-p)} \ < \ \infty.
\end{equation}
\end{lemma}
\begin{proof} Using the same arguments as in \cite[\S 10.5, Lemma 10.18]{SK} we conclude from (\ref{2c.1}) - (\ref{2c.2})
that for all $0 < \alpha <2$ and for $d=1,2$:
\begin{equation}\label{L1.1}
\hat a(p) + \hat a(-p)  = 2 - c_1 |p|^\alpha + o(|p|^\alpha) \quad \mbox{as } \; |p| \to 0
\end{equation}
where $c_1$ is some constant. In addition,  $\hat a(p) \in C(\R^{d}) $, $|\hat a(p)|<1$ for all $p \neq 0$,  and $\hat a(p) \to 0$ for $|p|\to \infty$. Thus $|1 -\hat a(p)|$ is separated from 0 for $|p|>\epsilon$ for any $\epsilon>0$. This implies that
\begin{equation}\label{L1.2}
\int\limits_{|p| \le 1} \frac{|\hat a(p)| \, dp}{2- \hat a(p)  - \hat a(-p)} \ < \ \infty.
\end{equation}
The convergence of the integral in (\ref{L1}) follows now from (\ref{L1.2}) and regularity condition (\ref{2b}) at infinity.
\end{proof}
Recall, that the equation \eqref{Last} for $n=2$ has form
\begin{equation}\label{13}
\hat L^{\ast}_2 k^{(2)} + f^{(2)}=0.
\end{equation}
Here
\begin{equation}\label{14}
f^{(2)}(x_1, x_2) \ = \ \varrho( a(x_1- x_2) + a(x_2- x_1)).
\end{equation}
and the operator $\hat L^{\ast}_2 \ = \  L^{(1)} + L^{(2)}$, where
\begin{equation}\label{15}
L^{(1)} k^{(2)}(x_1, x_2) \ = \ \int_{\R^{d}} a(x_1 - y)
k^{(2)}(y, x_2) dy - k^{(2)}(x_1, x_2),
\end{equation}
\begin{equation}\label{16}
L^{(2)} k^{(2)}(x_1, x_2) \ = \ \int_{\R^{d}} a(x_2 - y)
k^{(2)}(x_1, y) dy - k^{(2)}(x_1, x_2).
\end{equation}
Using the translation invariant property of the searched solution we get
$$
k^{(2)}(x_1, x_2) \ = \ k^{(2)}(x_1 - x_2, 0):=u^{(2)}(x_{1}-x_{2})
$$
In terms of the function $u^{(2)}$, the equation \eqref{13} has now the following form
\begin{align}
 \begin{aligned}\label{u2}
\int_{\R^{d}} a(x - y)
u^{(2)}(y) dy &+\int_{\R^{d}} a(y - x)
u^{(2)}(y) dy - 2u^{(2)}(x) \\
&=-\varrho( a(x) + a(- x)).
\end{aligned}
\end{align}
It is easy to check that the solution to this equation in the space $\BB(\X)$ always exists. Indeed, taking into account Lemma \ref{lem1}, we see at once that
\begin{equation}
u^{(2)}(x):=\frac{\varrho}{(2\pi)^{d}}\int_{\R^{d}}e^{i(p,x)}\frac{ \hat a (p) + \hat a (-p)}
{2 - \hat a(p) - \hat a (-p)}dp
\end{equation}
solves \eqref{u2} and it is an element of $\BB(\X)$.
Moreover, for any constant $A\in \R$ the function $$k^{(2)}(x_1, x_2):=u^{(2)}(x_{1}-x_{2})+ A$$  solves \eqref{13} and it is an element of $\XX_{2}$.
%
% After the Fourier transform we can rewrite (\ref{13}) - (\ref{16})
%as
%\begin{equation}\label{18}
%(\hat a(p) + \hat a (-p) -2) \ \hat k^{(2)}(p)  \ = \ -\varrho \
%(\hat a (p) + \hat a (-p)).
%\end{equation}
%Therefore,
%\begin{equation}\label{19}
%\hat k^{(2)}(p)  \ = \ \varrho \ \frac{ \hat a (p) + \hat a (-p)}
%{2 - \hat a(p) - \hat a (-p)} + A \delta(p),
%\end{equation}
%where $A$ is an arbitrary constant, and we will explain later how to
%choose $A$ in the general case.
%
%Using Lemma 1 we get that $\hat k^{(2)} (p)$ has an integrable
%singularity $\sim |p|^{-\alpha}$ at $p=0$.
%Thus there exist infinitely many translation invariant functions
%$k^{(2)}(x_1 - x_2) \in L^{\infty} (\R^{d})$ satisfying equation
%(\ref{13}).
%\\

Now let us turn to the general case. If for any $n>1$ we succeed to
solve equation (\ref{Last}) and express $k^{(n)}$ through
$f^{(n)}$, then knowing the expression of $f^{(n)}$ through
$k^{(n-1)}$ (see (\ref{f})), we get the solution $\{k^{(n)}\}_{n\geq 1}$ to the full system
(\ref{Last}) recurrently.
%It means that we have to invert the operator $\hat L_n^{\ast}$,
%and it is sufficient for us to do it on the class of
%translation invariant functions. The precise statement will be
%presented later for $(\hat L_n^{\ast})^{-1} f^{(n)}$, see formula
%(\ref{35}).
%
%Remind that
%\begin{equation}\label{20}
%\hat L_n^{\ast} \ = \ \sum_{i=1}^n L^i,
%\end{equation}
%where
%\begin{equation}\label{21}
%L^{i} k^{(n)}(x_1, \ldots, x_n) \ =
%\end{equation}
%$$
%\int\limits_{\R^{d}} a(x_i- y) k^{(n)}(x_1, \ldots, x_{i-1}, y, x_{i+1}, \ldots,
%x_n) dy - k^{(n)}(x_1, \ldots, x_n).
%$$
%\\

Recalling the definition of the norm on $\XX_{n}$, $n\geq 1$, we see that $|f| \leq |g|$ implies
$$
||f||\leq ||g||,\quad \text{for all}\quad  f, g \in \XX_{n}.
$$
These properties make the space $\XX_{n}$ a Banach lattice.
\begin{lemma} \label{3.2}
The operator $e^{t \hat L_n^{\ast}}$ is
positive on the Banach lattice $\XX_{n}$.
\end{lemma}
\begin{proof}
It is evident that
$$
A^i k^{(n)}(x_1, \ldots, x_n) \ := \ \int\limits_{\R^{d}} a(x_i- y) k^{(n)} (x_1, \ldots, x_{i-1}, y,
x_{i+1}, \ldots, x_n) dy.
$$
is positive and bounded on $\XX_{n}$ for any $1\leq i\leq n$. Taking into account
$$
\hat L_n^{\ast}  \ = \ \sum_{i=1}^n L^i,\quad
e^{t \hat L_n^{\ast}} \ = \ \otimes_{i=1}^n  e^{t L^{i}}, \quad e^{t
L^{i}} \ = \ e^{-t} e^{t A^{i}},
$$
where
\begin{align}
\begin{aligned}\label{21}
L^{i} k^{(n)}(x_1, \ldots, x_n) \ = &\int_{\R^{d}} a(x_i- y) k^{(n)}(x_1, \ldots, x_{i-1}, y, x_{i+1}, \ldots,
x_n) dy\\
& - k^{(n)}(x_1, \ldots, x_n).
\end{aligned}
\end{align}
we get the desired conclusion.
\end{proof}
%First consider the restriction of $\hat L_n^{\ast}$ to the invariant
%subspace consisting of the functions of the form
%$$
%\varphi(\tau(x_1), \ldots, \tau(x_n)) \prod_{i=1}^n q(\sigma(x_i)),
%\quad \mbox{ where } \; \varphi (w_1, \ldots, w_n) \in
%L^{\infty}_{inv}((R^n)^d).
%$$
%The operator $\hat L_n^{\ast}$ acts in $\XX_n$ as
%\begin{equation}\label{24}
%\hat L_n^{\ast}  \ = \ \sum_{i=1}^n L^i,
%\end{equation}
%where
%\begin{equation}\label{25}
%L^{i} \ \varphi (w_1, \ldots, w_n) \   =
%\ \int\limits_{\R^{d}} a(w_i-u)
%\varphi(w_1, \ldots, w_{i-1}, u, w_{i+1}, \ldots, w_n) du - \varphi
%(w_1, \ldots, w_n).
%\end{equation}
%This formula follows from the equality $Qq=q$. Remind that
%$\kappa_{cr}$ is "absorbed" in $Q$. Formula (\ref{25}) means that in
%this case we have only spatial convolutions and no integration over $S$.
%In the Fourier variables the operator $\hat L_n^{\ast}$ acts as a
%multiplication operator by the function
%$$
%\sum_{i=1}^n \hat a (p_i) \ - \ n.
%$$
%To invert $\hat L_n^{\ast}$ let us notice that if $\varphi (w_1, \ldots,
%w_n)$ is a translation invariant function then its Fourier transform
%has a form
%$$
%\hat \varphi (p_1, \ldots, p_n) \ \delta (p_1+ \ldots + p_n).
%$$
%On the subspace of the "momentum space" $(p_1, \ldots, p_n)$
%specified by the equation $p_1+ \ldots + p_n = 0$ the function $
%\frac{1}{\sum_{i=1}^n \hat a (p_i) - n} $ has an integrable
%singularity  at $p=0$.

Next we will construct a solution to the system (\ref{Last}) satisfying \eqref{Th1}-\eqref{estimate}.
% and
%\begin{equation}\label{est}
%k^{(n)} (x_1, \ldots, x_n) \ \le \ K_n,
%\end{equation}
%where $K_n = D C^n (n!)^2$, $D,\ C$ are constants.
As follows from (\ref{f}), the function $f^{(n)}$ is the sum of
functions of the form
\begin{equation}\label{32}
f_{i,j} (x_1, \ldots, x_n)  =  k^{(n-1)} (x_1,\ldots,\check{x_i},
\ldots, x_n)  a(x_i - x_j), \quad i\neq j.
\end{equation}
%Below we invert the operator $\hat L_n^{\ast}$ on the set of
%functions of the form (\ref{32}).
%%, see (\ref{52}) below.
%
%Set \begin{equation}\label{35}
%v^{(n)}_{i,j}  \ =  \ \int_0^{\infty} e^{t \hat L_n^{\ast}} f_{i,j} \ dt.
%\end{equation}
%where $f$ is a function of the form (\ref{32}),
%As a result,
%\begin{equation}\label{35A}
%k^{(n)}  \ =  \ \int\limits_0^{\infty} e^{t \hat L_n^{\ast}} f^{(n)} \ dt \ = \ \sum_{i \neq j} v^{(n)}_{i,j}.
%\end{equation}
We suppose by induction that
$$
k^{(n-1)} (x_1, \ldots, x_{n-1}) \ \le \ K_{n-1}, \quad \text{for all } \; (x_1, \ldots, x_{n-1})\in(\X)^{n-1},\quad n\geq 2,
$$
where $K_n = D C^n (n!)^2$, and $D, C$ are some constants. Consequently,
\begin{equation}\label{34}
f_{i,j}(x_1, \ldots, x_n) \ \le \  K_{n-1} a(x_i - x_j),\quad (x_1, \ldots, x_{n})\in(\X)^{n-1}.
\end{equation}
%Since the function $q(s)$ is strictly positive on the compact $S$,
%the following inequality holds:
%\begin{equation}\label{aq}
%a(x_i, x_j) \ \le \ c \ q(\sigma(x_i)) \ \alpha(\tau(x_i) -
%\tau(x_j))
%\end{equation}
%with a constant $c$.
Using the positivity of the operator $e^{t \hat L_n^{\ast}}$ and (\ref{34}) we have
%(\ref{20}),
%(\ref{24}) and
\begin{align}
\begin{aligned}\label{36}
\left(e^{t \hat L_n^{\ast}} f_{i,j} \right) (x_1, \ldots, x_{n})\ \le \  K_{n-1} \ \left(e^{t \hat L_n^{\ast}}
a(\cdot_i - \cdot_j) \right)(x_1, \ldots, x_{n}).
\end{aligned}
\end{align}
An easy observation $e^{t L^{i}}1\!\!1=1\!\!1$, $\forall i=1,\,\ldots, n,$
 $1\!\!1(x)\equiv 1$, shows
\begin{equation}\label{36_1}
\left(e^{t \hat L_n^{\ast}}
a(\cdot_i - \cdot_j)\right) (x_1, \ldots, x_{n})\  =  \
\left(e^{t ( L^i + L^j)} a(\cdot_i - \cdot_j)\right) (x_1, \ldots, x_{n}).
\end{equation}
Note that the latter function depends only on variables $x_{i}$ and $x_{j}$. Indeed, it follows from the identity
$$
[(L^i + L^j) a(\cdot_i - \cdot_j)](x_1, \ldots, x_{n})=(a\ast a)(x_{i}-x_{j}) +(a\ast a)(x_{j}-x_{i})-2a(x_{i}-x_{j}),
$$
where the symbol $\ast$ stands for the convolution of functions.

Set
$$
L\!\!\!L v (x):= (a\ast v)(x) +(a\ast v)(-x)-2v(x).
$$
Clearly, the operator $L\!\!\!L$ is bounded in $L^{1}(\R^{d})$ and if, additionally, $\hat v\in L^{1}(\R^{d})$ then
\begin{equation}
e^{tL\!\!\!L}v(x)=\frac{1}{(2 \pi)^d}\int_{\X} e^{i(p,\,x)} e^{t (\hat a (p) + \hat a (-p) -2)} \hat v(p)dp.\label{LL}
\end{equation}
Moreover,
$$
[(L^i + L^j) a(\cdot_i - \cdot_j)](x_1, \ldots, x_{n})=L\!\!\!L  a(x_{i}-x_{j})
$$
and thus
\begin{equation}
e^{t ( L^i + L^j)} a(\cdot_i - \cdot_j) (x_1, \ldots, x_{n})=e^{tL\!\!\!L}a(x_{i}-x_{j}).\label{ij}
\end{equation}
%$$
%c K_{n-1}  e^{t ( L^i_{max} + L^j_{max})} \alpha(\tau(x_i) -
%\tau(x_j)) \  \prod_{l=1}^n  q(\sigma(x_l)).
%$$
Using \eqref{2a}-\eqref{2b}, \eqref{less1}, and  \eqref{36}-\eqref{ij}
%the fact $a\in L^{1}(\R^{d})$ and
%$$
%e^{t ( L^i + L^j)} a(\cdot_i - \cdot_j) (x_1, \ldots, x_{n})=e^{tL\!\!\!L}a(x_{i}-x_{j})
%$$
we conclude
\begin{equation}\label{36_3}
||e^{t \hat L_n^{\ast}} f_{i,j}||_{\XX_{n}}\leq  \frac{K_{n-1}}{(2 \pi)^d}\int_{\R^{d}}  e^{t (\hat a (p) + \hat a (-p) -2)} |\hat a(p)|
dp\to 0,\quad t\to\infty.
\end{equation}
The latter convergence is due to the Lebesgue dominated convergence theorem.

We next show that $e^{t \hat L_n^{\ast}} f_{i,j}$ is integrable with respect to $t$ on $\R_{+}$. According to \eqref{36_3}
%formula (\ref{25}), the Fourier transform and the Fubini
%theorem we finally obtain from (\ref{2b}) and (\ref{36}) the upper
%bound on $v^{(n)}_{i,j}$:
\begin{equation}\label{vij}
v^{(n)}_{i,j} \ = \ \int_0^{\infty} e^{t \hat L_n^{\ast}} f_{i,j} \ dt \
\le
% K_{n-1}   \int\limits_0^{\infty}
%\int\limits_{\R^{d}} e^{t (\hat a (p) + \hat a (-p) -2)} |\hat
%a(p)| dp dt  \ = \
 \frac{Z K_{n-1}}{(2 \pi)^d},
\end{equation}
where
\begin{align}
\begin{aligned} \label{A}
Z  = \int_0^{\infty} \int_{\R^{d} } e^{t (\hat a (p) + \hat a (-p) -2)} |\hat a(p)|
dp dt   \le \int_{\R^{d}} \frac{|\hat a(p)|}{2- \hat a
(p) - \hat a (-p)} dp < \infty
\end{aligned}
\end{align}
due to the Fubini theorem and Lemma \ref{lem1}.

Our next goal is to show that $$k^{(n)}  = \sum_{i \neq j} v^{(n)}_{i,j}=\int_0^{\infty} e^{t \hat L_n^{\ast}} f^{(n)}$$ is a solution to \eqref{Last} in $\XX_{n}$.
It is easily seen from \eqref{vij} and induction procedure that $k^{(n)}\in\XX_{n}$. Since $e^{t \hat L_n^{\ast}}$ is a strongly continuous semigroup we have
$$
e^{t \hat L_n^{\ast}}f^{(n)}-f^{(n)}=\hat L_n^{\ast}\int_{0}^{t}e^{s \hat L_n^{\ast}}f^{(n)}ds.
$$
A passage to the limit as $t\to\infty$ together with  \eqref{36_3} shows that $k^{(n)}$ is a solution to \eqref{Last} in $\XX_{n}$.

Since the function $f^{(n)}$ is the sum
of functions $f_{i,j}$, $i\neq j$ we deduce that $ k^{(n)}$
%given by
%$$
%k^{(n)}  =
%%\ \left( - \hat L_n^{\ast} \right)^{-1} f^{(n)} \ = \
%\sum_{i \neq j} v^{(n)}_{i,j}
%$$
is bounded by $C n^2 K_{n-1}$
for some $C>0$. Thus we get the recurrence inequality
\begin{equation}\label{50}
K_n \ \le \ C n^2 K_{n-1},
\end{equation}
and by induction it follows that
\begin{equation}\label{49}
K_n \ \le \ C^n \, (n!)^2.
\end{equation}
Thus
\begin{equation}\label{49A}
k^{(n)} (x_1, \ldots, x_n) \ \le \ C^n \, (n!)^2.
\end{equation}
Moreover, using \eqref{36}-\eqref{ij} we have
%the positivity of $f(x_1, \ldots, x_n)$ (see (\ref{32})), inequality (\ref{36}), the Fourier transform and the Fubini
%theorem as above, we get from (\ref{35})
\begin{align}
\begin{aligned}
v^{(n)}_{i,j}(x_1, \ldots, x_n)  & =   \int_0^{\infty} \left( e^{t \hat L_n^{\ast}} f_{i,j} \right) (x_1, \ldots, x_n) dt \\
&\le
\frac{K_{n-1}}{(2 \pi)^d}
\int_0^{\infty} \int_{\R^{d}}  e^{t (\hat a (p) + \hat a (-p) -2)}{\hat a(p) \ e^{ip(x_i - x_j) }} dpdt.
\end{aligned}
\end{align}
Integrability of the function $e^{t (\hat a (p) + \hat a (-p) -2)}{|\hat a(p)|}$ and the
Lebesgue-Riemann lemma imply that the function $v^{(n)}_{i,j}$ satisfies the following condition:
\begin{equation}\label{38A}
v^{(n)}_{i,j}(x_1, \ldots, x_n) \ \to \ 0 \quad \mbox{ if } \;
|x_i - x_j|  \to \infty.
\end{equation}
%for all $i \neq j$.
Consequently,
\begin{equation}\label{100}
%\left( - \hat L_n^{\ast} \right)^{-1} f^{(n)} (x_1, \ldots, x_n)
k^{(n)} (x_1, \ldots, x_n) \ = \
\sum_{i \neq j} v^{(n)}_{i,j}  (x_1, \ldots, x_n) \ \to \ 0,
\end{equation}
whenever $|x_i - x_j|  \to \infty$ for all $i \neq j$.
Eventually, we have constructed $\{k^{(n)}\}_{n\geq 1}$ satisfying estimate (\ref{49A}) and condition (\ref{100}).

Of course, any functions of the form
$$
k^{(1)}\equiv \varrho,\quad k^{(n)} (x_1, \ldots, x_n) \ = \ \int\limits_0^{\infty} e^{t \hat
L_n^{\ast}} f^{(n)} (x_1, \ldots, x_n) \ dt \ +  \ A_n,\quad n \geq 2
$$
where $A_n$ are arbitrary constants, are solution to the system
(\ref{Last}) too. Among different $A_{n}$ we have to find such constants $A_{n}$ for which
\begin{equation}\label{51}
| k^{(n)} (x_1, \ldots, x_n) \ - \ \varrho^n| \ \to \ 0, \; \mbox{ whenever }  |x_i - x_j|
\to \infty
\end{equation}
for all $i\neq j$. Taking $A_n=\varrho^n$ we conclude that
\begin{equation}\label{52}
k^{(1)}_\varrho\equiv\varrho,\quad k^{(n)}_\varrho \ = \  \int\limits_0^{\infty} e^{t \hat L_n^{\ast}} f^{(n)}
dt \ + \ \varrho^n,\quad n\geq 2
\end{equation}
is the wanted solution to \eqref{Last} in the Banach spaces $(\XX_{n})_{n\geq 1}$. To emphasize the dependence of $f^{(n)}$ on $\varrho$, we will sometimes use notation $f_{\varrho}^{(n)}$  for $f^{(n)}$. It is clear that the last term in (\ref{52}) vanishes under the
action of $\hat L_n^{\ast}$, and (\ref{51}) holds because of
\eqref{38A}.
%Denote this solution by $k_\varrho^{(n)}$.
For the solutions $\{ k_{\varrho}^{(n)} \}_{n\geq 1}$ instead of (\ref{50}) we have the recurrence
\begin{equation}\label{53}
K_n \ \le \ C n^2 K_{n-1} \ + \ \varrho^n,
\end{equation}
%Taking $L_n=\frac{K_n}{C^n (n!)^2}$ we have
%$$
%L_n  \ \le \ L_{n-1} \ + \ \frac{\varrho^n}{C^n (n!)^2} \ \le \ D
%$$
%with some positive constant $D>0$. Thus we have
which yields
\begin{equation}\label{55}
K_n \ \le \ D C^n (n!)^2.
\end{equation}
%which differs from (\ref{49}) only by the constant factor.

%Thus we proved the existence of solution $\{ k_{\varrho}^{(n)} \}$
%of the system (\ref{Last}) corresponding to the stationary problem.
To be certain that the constructed system $\{ k_{\varrho}^{(n)} \}_{n\geq 1}$ is a system of correlation functions, i.e., it corresponds to a probability measure $\mu^\varrho$ on the configuration space $\Gamma$, we will
prove in the next section that $\{ k_{\varrho}^{(n)} \}_{n\geq 1}$ can be constructed as the limit when $t \to \infty$ of the system
of correlation functions $\{ k_{t}^{(n)}\}_{n\geq 1}$ associated with the solution to the Cauchy problem \eqref{59} with
corresponding initial data from $\K$.

\section{The proof of Theorem \ref{mainth} (ii).  }

In this section we find the solution to the Cauchy problem
(\ref{59}), (\ref{k0-1}) - (\ref{rij}) and prove the relation (\ref{Th1-2}) using the method of mathematical induction.
By the variation of parameters formula we have
\begin{equation}\label{61}
k_{t}^{(n)} \ = \ e^{t \hat L_n^{\ast}} k_{0}^{(n)} \ + \  \int\limits_0^t e^{(t-s) \hat
L_n^{\ast}} f_s^{(n)} \ ds,
\end{equation}
where $f_s^{(n)}$ is expressed through $k_s^{(n-1)}$ by
(\ref{f}). Using the identity $$ \hat L_n^{\ast}
k_\varrho^{(n)} \ = \ - f_\varrho^{(n)},$$
where
$$
f_\varrho^{(n)}(x_1, \ldots, x_n) \ = \ \sum_{i,j:\ i\neq j}
k_\varrho^{(n-1)}(x_1, \ldots,\check{x_i}, \ldots, x_n) \ a(x_i-x_j),
$$
we get
$$
\left( e^{t \hat L_n^{\ast}} - E \right) k_\varrho^{(n)} \ = \ -
\int\limits_0^t \frac{d}{ds} e^{(t-s) \hat L_n^{\ast}} k_\varrho^{(n)} ds \
\ = \ - \int\limits_0^t e^{(t-s) \hat L_n^{\ast}}  f_\varrho^{(n)} \ ds,
$$
and therefore
%$$
%k^{(n)}(t) -  k_\varrho^{(n)} \ = \ \left( e^{t \hat L_n^{\ast}} - E
%\right) k_\varrho^{(n)} \ +
%$$
%$$
%e^{t \hat L_n^{\ast}}(k^{(n)}(0) - k_\varrho^{(n)}) \ + \  \int\limits_0^t
%e^{(t-s) \hat L_n^{\ast}} f^{(n)}(s) \ ds \ =
%$$
\begin{equation}\label{64}
k_{t}^{(n)} -  k_\varrho^{(n)} \ = \
e^{t \hat L_n^{\ast}}(k_{0}^{(n)} - k_\varrho^{(n)}) \ + \  \int\limits_0^t
e^{(t-s) \hat L_n^{\ast}} (f_{s}^{(n)} -  f_\varrho^{(n)}) \ ds.
\end{equation}
%Here $f_\varrho^{(n)}$ are expressed in terms of $k_\varrho^{(n-1)}$
%by (\ref{f}).
%and we used that the equation $ \hat L_n^{\ast}
%k_\varrho^{(n)} \ = \ - f_\varrho^{(n)}$ implies
%$$
%\left( e^{t \hat L_n^{\ast}} - E \right) k_\varrho^{(n)} \ = \ -
%\int\limits_0^t \frac{d}{ds} e^{(t-s) \hat L_n^{\ast}} k_\varrho^{(n)} ds \
%\ = \ - \int\limits_0^t e^{(t-s) \hat L_n^{\ast}}  f_\varrho^{(n)} \ ds.
%$$
We will prove now that both terms in the right-hand side of (\ref{64}) converge to 0 in
the norm of $\XX_n$.

The identity \eqref{k0} and the inversion formula (\ref{52})
yield
\begin{equation}\label{65}
e^{t \hat L_n^{\ast}} \big( k_{0}^{(n)} - k_\varrho^{(n)} \big) \ = \ e^{t \hat
L_n^{\ast}}\big( r^{(n)} - v^{(n)} \big)  \ = \  e^{t \hat
L_n^{\ast}} r^{(n)} - e^{t \hat L_n^{\ast}} v^{(n)},
\end{equation}
where
\begin{equation}\label{66}
v^{(n)} \ = \ \int_0^{\infty} e^{s \hat L_n^{\ast}} f_\varrho^{(n)}
\ ds.
\end{equation}
We consider  terms $e^{t \hat L_n^{\ast}} v^{(n)}$ and $ e^{t \hat
L_n^{\ast}} r^{(n)}$  in (\ref{65}) separately.

The first term can be estimated using the inequality in \eqref{36_3} and bound \eqref{55}. As a result, we get
$$
\left| \left( e^{t \hat L_n^{\ast}} \ v^{(n)} \right) (x_1, \ldots,
x_n) \right|
$$
$$
 \le \frac{D C^{n-1} ((n-1)!)^2}{(2 \pi)^d} \ \sum_{i,j:\, i\neq j} \int_{\R^{d}}  \int_{0}^{\infty}e^{(t+s) (\hat a (p) + \hat a (-p) -2)} |\hat a(p)|dsdp
$$
$$
 \le \frac{D C^{n-1} (n!)^2}{(2 \pi)^d} \ \int_{\R^{d}} e^{t (\hat a (p) + \hat a (-p) -2)} \frac{|\hat a(p)|}{2-\hat a (p) - \hat a (-p)}dp
$$
Due to the Lebesgue dominated convergence theorem and Lemma \ref{lem1} the latter integral converges to zero as $t$ tends to $\infty$. Hence, $$||e^{t \hat L_n^{\ast}} \ v^{(n)}||_{{\XX}_n}\to 0,\quad t\to\infty.$$

The second term  $ e^{t \hat L_n^{\ast}} r^{(n)}$  in (\ref{65}) can be handled in much the same way, if instead of
$\frac{|\hat a (p)|}{2 - \hat a (p) - \hat a (- p)} $ we consider
$$
\frac{|\hat r^{(n)}_{ij} (p)|}{2 - \hat
a (p) - \hat a (- p)} \quad \mbox{with } \; \hat r^{(n)}_{ij} (p) \in L^1(\R^{d}).
$$

Our next goal is to show that
\begin{equation}\label{81}
\int\limits_0^t e^{(t-s) \hat L_n^{\ast}} (f_{s}^{(n)} - f_\varrho^{(n)}) \
ds \ \to 0
\end{equation}
in sup-norm when $t \to \infty$. To this end we use the induction with respect to  $n$.
The base case of induction is satisfied since
\begin{equation}\label{82}
k_{0}^{(1)}(x) \ = \   k_{t}^{(1)}(x) \ \equiv \ k_\varrho^{(1)}(x) \ = \ \varrho.
\end{equation}
Let us assume induction step
\begin{equation}\label{ind}
\| k_{t}^{(n-1)} \ - \  k_\varrho^{(n-1)} \|_{\XX_{n-1}} \ \to \ 0
\quad \mbox{ as} \quad t \to \infty,\quad n>2.
\end{equation}
The latter convergence implies that
\begin{equation}\label{83}
\| k_{t}^{(n-1)} \|_{\XX_{n-1}}  \ \le \  M_{n-1} \quad \mbox{ for all
} \; t \ge 0
\end{equation}
with some positive constant depending only on $n$. Indeed, the
operator $\hat L_n^{\ast}$ is bounded and the function $a$ is
bounded, hence the norm of the solution $k_{t}^{(n)}$ of the problem
(\ref{59}) (with any bounded for $l \le n$ initial data) is
evidently bounded on any compact time interval $[0,\tau]$. On the
other hand, for any $\varepsilon >0$ there exists $\tau$ such that
for all $t > \tau$ the norm $\| k_{t}^{(n-1)} - k_\varrho^{(n-1)}\|_{\XX_{n-1}}
<\varepsilon$ by (\ref{ind}). Thus the bound (\ref{83}) is proved.

From (\ref{ind}) and (\ref{f}) it follows that
\begin{equation}\label{83A}
\| f_{t}^{(n)} \ - \  f_\varrho^{(n)}\|_{\XX_n} \ \to 0 \quad \mbox{
as} \; t \to \infty.
\end{equation}
To estimate the integral (\ref{81}) we rewrite it as a sum
\begin{equation}\label{84}
\left( \int\limits_0^\tau \ + \ \int\limits_\tau^t \right) e^{ s \hat L_n^{\ast}}
(f_{t-s}^{(n)} - f_\varrho^{(n)}) \ ds.
\end{equation}
Let us estimate the second integral in (\ref{84}) using Lemma \ref{3.2}:
\begin{align}
\begin{aligned}\label{85}
\left| \int_\tau^t  e^{s \hat L_n^{\ast}} (f_{t-s}^{(n)} -
f_\varrho^{(n)}) ds \right| & \le  \int_\tau^t  e^{s \hat
L_n^{\ast}} ( |f_{t-s}^{(n)}| + |f_\varrho^{(n)}| ) \ ds \\
& \le ( M_{n-1}  +  \| k_\varrho^{(n-1)} \|_{\XX_{n-1}} )  \int_\tau^t
e^{s \hat L_n^{\ast}} \sum_{i \neq j} a(\cdot_i - \cdot_j) \ ds.
\end{aligned}
\end{align}
Then using the same arguments as above we conclude that it will be sufficient to estimate for any pair $i
\neq j$ the following integral
\begin{equation}\label{89}
\int\limits_\tau^{t} \int\limits_{\R^{d}} e^{s(\hat a(p) + \hat a (-p) -2)}
|\hat a (p)| \ dp \ ds \ \le \  \int\limits_\tau^{\infty} \int\limits_{\R^{d}}
e^{s(\hat a(p) + \hat a (-p) -2)} |\hat a (p)| \ dp \
ds
\end{equation}

Since the integral
\begin{equation}\label{90}
\int\limits_0^{\infty} \int\limits_{\R^{d}} e^{s(\hat a (p) + \hat a (-p)
-2)} |\hat a (p)| \ dp \ ds \ = \ \int\limits_{\R^{d}} \frac{|\hat a
(p)|}{2 - \hat a (p) - \hat a (- p)} \ dp
\end{equation}
converges, then the integral (\ref{89}) tends to 0 when $\tau \to
\infty$. Consequently we can take $\tau$ in such a way that
(\ref{89}) is less than $\varepsilon$, and then (\ref{85}) is less
than $C_n \varepsilon$ for some $C_n$ and any $t>\tau$.

Finally let us estimate the first integral in (\ref{84}) for a given
$\tau$:
\begin{equation}\label{91}
\int\limits_0^\tau  e^{ s \hat L_n^{\ast}} (f_{t-s}^{(n)} - f_\varrho^{(n)})
\ ds.
\end{equation}
From (\ref{83A}) it follows that we can choose $t_0>\tau$ such that
for $t> t_0$ the following estimate holds
$$
\| f_{t-\tau}^{(n)} - f_\varrho^{(n)} \|_{\XX_n} \ < \
\frac{\varepsilon}{\tau}.
$$
Consequently the norm of (\ref{91}) is less than $\varepsilon$.
Finally, for $t>t_0$ the integral in (\ref{81}) is less than
$(C_n + 1) \varepsilon$ in sup-norm and convergence (\ref{81}) as well as
(\ref{Th1-2}) is proved.

\medskip

Thus we proved the strong convergence (\ref{Th1-2}), and the proof of the second part of Theorem 1 is completed.

The final step of the proof of the first part of Theorem 1 is to show that the system of correlation functions $\{k_\varrho^{(n)} \}$ corresponds to a probability measure $\mu^\varrho$ on the configuration space $\Gamma$. We will use here the following Proposition summarizing results of two papers \cite{L1} and \cite{L2} of A. Lenard.

\medskip

\begin{proposition}(see \cite{L1}, \cite{L2}) \label{propos2} If the system of correlation functions $\{k^{(n)} \}$ satisfies Lenard positivity and moment growth conditions then there exists a unique probability measure $\mu \in {\cal M}^1_{fm}(\Gamma)$, locally absolutely continuous with respect to a Poisson measure, whose correlation functions are exactly $\{k^{(n)} \}$.
\end{proposition}
\medskip

For the convenience of the reader we formulate these conditions below.

\noindent {\it Lenard positivity.} $KG \ge 0$ for any $G \in B_{bs}(\Gamma_0)$ implies
\begin{equation}\label{LP}
\sum_{n=0}^{\infty} \frac{1}{n!} \int\limits_{\R^{d}} \ldots \int\limits_{\R^{d}} G^{(n)}(x_1, \ldots, x_n) k^{(n)}(x_1, \ldots, x_n) dx_1 \ldots dx_n \ge 0.
\end{equation}

\medskip

\noindent {\it Moment growth.} For any bounded set $\Lambda \subset \R^{d}$ and $j \ge 0$
\begin{equation}\label{MG}
\sum_{n=0}^{\infty} (m^\Lambda_{n+j})^{- \frac{1}{n}} \ = \ \infty,
\end{equation}
where
$$
m^\Lambda_{n} = (n!)^{-1} \int\limits_\Lambda \ldots \int\limits_{\Lambda} k^{(n)}(x_1, \ldots, x_n) dx_1 \ldots dx_n.
$$

\medskip

In our case it follows from  (\ref{55}) that  $\big( m_n^\Lambda \big)^{-\frac1n} \ge \frac{\tilde C}{n}$. Thus condition (\ref{MG}) of the uniqueness holds.
%Using results from \cite{KS} and \cite{KKP} (see Proposition in \cite[Appendix]{KKP}), we can conclude that for any $t>0$ the solution %$\{k^{(n)}(t) \}$ of the Cauchy problem (\ref{59}), (\ref{k0-1})-(\ref{rij}) satisfies condition (\ref{LP}).
%That implies, see  \cite[Th. 4.1]{L2}, the existence of a probability measure $\mu \in {\cal M}^1_{fm}(\Gamma)$ which is locally %absolutely continuous with respect to a Poisson measure and has correlation functions  $\{k^{(n)}(t) \}$. In addition, in our case the %moment growth condition (\ref{MG}) holds.
%Using results of \cite{L1} we conclude that this measure is unique.
To obtain the Lenard positivity condition (\ref{LP}) we use that $\{ k_{\varrho}^{(n)} \}$ was constructed as the limit when $t \to \infty$ of the system of correlation functions $\{ k_{t}^{(n)} \}$
associated with the solutions of the Cauchy problem (\ref{59}) with initial data satisfying (\ref{k0-1}) - (\ref{rij}) and corresponding to some measure  $\mu_0 \in {\cal M}^1_{\rm{corr}}(\Gamma)$ (e.g. the Poisson measure):
\begin{equation}\label{limk}
k^{(n)}_\varrho \ = \ \lim_{t\to\infty} k_{t}^{(n)}.
\end{equation}
Using results from \cite{KKP} (Proposition 4.4 and Corollary 4.1) we can conclude that for any $t>0$ the solution $\{
k_{t}^{(n)} \}$ of the Cauchy problem (\ref{59}), (\ref{k0-1})-(\ref{rij}) satisfies condition (\ref{LP}) of Lenard positivity, see Appendix for the detailed proof of this important statement.
Consequently, the limit system of correlation functions $k^{(n)}_\varrho $
also satisfies the Lenard positivity condition (\ref{LP}).

Thus Proposition \ref{propos2} implies that there exists a unique probability measure $\mu^\varrho \in {\cal M}^1_{\rm{corr}}(\Gamma)$, locally absolutely continuous with respect to a Poisson measure $\pi_\varrho$, whose correlation functions are $\{k^{(n)}_\varrho \}$.
This completed the proof of Theorem \ref{mainth}.

\section {Concluding remarks.}

\begin{remark} {\bf  Law of large numbers.} Theorem \ref{mainth} implies the law of large numbers for the number of particles, i.e. the
existence of the spatial density of particles:
$$
\frac{N(V)}{|V|} \ \to \ \varrho, \quad \mbox{ as } \ |V| \to \infty,
$$
with convergence in probability.
\end{remark}
\begin{proof} Let us define the random variable
$$
N(V) (\gamma) = \sum_{x \in \gamma} \chi\big._V  (x), \quad \chi\big._V \; \mbox{is the  characteristic function of } \; V,
$$
equals to the number of particles in the domain $V\subset \R^{d}$. By the definition of correlation functions we have
$$
\mathbb{E} [N(V)] = \int\limits_{\R^{d}} \chi_V (x) k_\varrho^{(1)}(x) dx = \int\limits_V  k_\varrho^{(1)}(x) dx = \varrho |V|, \quad \mbox{where } \; |V|= vol \, V.
$$
Analogously,
$$
\mathbb{E} [N(V)(N(V) - 1)] = \int\limits_{V} \int\limits_{V}  k_\varrho^{(2)}(x_1, x_2) dx_1 dx_2.
$$
So the variance of $N(V)$ equals to
\begin{equation}\label{D1}
Var ( N(V)) = \int\limits_{V} \int\limits_{V}  k_\varrho^{(2)}(x_1, x_2) dx_1 dx_2 + \varrho |V| - \big( \varrho |V| \big)^2,
\end{equation}
and
\begin{equation}\label{D2}
 Var \Big(\frac{N(V)}{|V|} \Big) = \frac{1}{|V|^2} \int\limits_{V} \int\limits_{V} \big( k_\varrho^{(2)}(x_1, x_2) - \varrho^2 \big) dx_1 dx_2 + \frac{\varrho}{ |V|}.
\end{equation}
Theorem \ref{mainth} implies the correlation decay
$$
| k^{(2)}_\varrho (x_1, x_2) \ - \ \varrho^2 | \ \to \ 0,
$$
when  $|x_1 - x_2| \to \infty$.
Consequently both terms in (\ref{D2}) tend to 0 when $|V| \to \infty$, and
the direct application of the Chebyshev inequality gives the convergence of $\frac{N(V)}{|V|}$  to $ \varrho$ in probability.
\end{proof}
\medskip

%{\bf Remark 2.}  The initial system of correlation functions $k^{(n)}(0)$ given by (\ref{k0-1})-(\ref{rij}) describes for example the Gibbs measure that is a
%Gibbs reconstruction of the Poisson field with intensity $\varrho$ by the pair stable regular potential in high temperature region,
%see e.g. \cite[4.4.7]{R}.

%The similar estimates for correlation functions were obtained for low density gases in \cite{DS}

\begin{remark}
We can include a possibility for particles to jump. The analogous model has been considered earlier in \cite{KKS}.
More precisely, let us consider the following heuristic generator
$L\ + \ L_J$, where $L$  was defined by (\ref{generator}) and
$$
L_J F(\gamma)\ = \  \int\limits_{\R^{d}} \sum_{x\in \gamma}
J(x-y) \Big( F ((\gamma\setminus x) \cup y)-F(\gamma) \Big) \, dy
$$
Then in Lemma \ref{lem1} it will appear in the denominator
$$
2 \ - \ \hat{a}(p) \ - \ \hat{a}(-p) \  + \ \hat J (0) \ -  \ \hat{J}(p), \quad  \hat J (0) =  \int\limits_{\R^{d}} J(u) du,
$$
instead of $2 - \hat{a}(p) -  \hat{a}(-p) $, and, hence, integrability \eqref{L1} for any $a$ is satisfied provided the jump
kernel $J$ has heavy tails.

The interpretation of this effect is the following: if individuals of a population have
possibility to choose between breeding and emigration far from their homeland, then in the critical regime the density of clusters
is decaying enough establishing regular pair correlation.
\end{remark}

%There is a biological model of bees population in such spirit.
%Most of bees are moving not far from the areal. But small part
%of them have quite long interventions. It means that we can have $J= J_1 + J_2$ where $J_1$ is short range and $J_2$
%may be very small but long range, i.e. very small part of population has high mobility. Then the creation of an equilibrium state is %completely defined by this active sub-population.

\medskip

\begin{remark}\label{Remark3} Let $k_{0}^{(n)}(x_1, \ldots, x_n)= \varrho^n$, and one of the following two conditions holds: \\
$$
\text{{\bf A1)}}\quad a(x) \sim \frac{1}{|x|^{\alpha+1}} \quad \mbox{as } \; |x| \to \infty \; \mbox{ with } \; 1 \le \alpha < 2 \quad \mbox{in the case } \quad d=1,\quad
$$
\medskip
or \\
{\bf A2)} $\int\limits_{\R^{d}} |x|^2 a(x) dx < \infty$ in the case $d=1,2$.

Then
\begin{equation}\label{P1}
k_{t}^{(2)}(0,0) \to \infty \quad \mbox{as } \; t \to \infty.
\end{equation}
\end{remark}
\begin{proof} Using (\ref{f}) we have for any $t \ge 0$:
$$
f_{t}^{(2)}(x_1, x_2) = \varrho (a(x_1 - x_2) + a(x_2 - x_1)).
$$
Since the operator $L_n^{\ast}$ annihilate constants, then we get from (\ref{61})
\begin{align}\label{P1-61}
\begin{aligned}
k_{t}^{(2)}(x_1 - x_2, 0)  =  e^{t \hat L_2^{\ast}} \varrho^2  +   \int\limits_0^t e^{(t-s) \hat
L_2^{\ast}} f_{s}^{(2)}(x_{1},x_{2})  ds \\
 = \varrho \int\limits_0^t e^{s \hat
L_2^{\ast}}   (a(\cdot_1 - \cdot_2) + a(\cdot_2 - \cdot_1))(x_{1},x_{2}) ds + \varrho^2.
\end{aligned}
\end{align}
If $x_1 = x_2$, then
$$
k_{t}^{(2)}(0, 0)  = \varrho \int\limits_0^{t} \int\limits_{\R^{d}} e^{s (\hat a (p) + \hat a (-p) -2)} \
 (\hat a (p) + \hat a(-p)) ds dp + \varrho^2.
$$
If condition {\bf A1} is fulfilled, then decomposition (\ref{L1.1}) from the proof of Lemma \ref{lem1} implies
$$
\int\limits_{\R^{d}} e^{s (\hat a (p) + \hat a (-p) -2)} \
 (\hat a (p) + \hat a(-p)) dp \sim s^{-1/\alpha} \quad \mbox{as } \; s \to \infty,
$$
and since $\int\limits_0^t s^{-1/\alpha} ds \to \infty \; (t \to \infty)$ for $1\le \alpha<2$, then we obtain (\ref{P1}).

Under condition {\bf A2} we get
$$
\hat a(p) + \hat a(-p) = 2 - c_1|p|^2 + o(|p|^2) \quad \mbox{as } \; |p| \to 0,
$$
and
$$
\int\limits_{\R^{d}} e^{s (\hat a (p) + \hat a (-p) -2)} \
 (\hat a (p) + \hat a(-p)) dp \sim s^{-d/2} \quad \mbox{as } \; s \to \infty.
$$
Consequently in the case $d=1,2$ we again obtain (\ref{P1}). The similar estimates yield that $k_{t}^{(2)}(0, x) \to \infty$ for any $x$ as $t \to \infty$.
\end{proof}

The growth (in $n$) of correlation functions shows a presence
of strong clustering in the system. For dispersal kernels
with short range these clusters become so dense that
the pair correlation function grows to infinity as $t \to \infty$.

\section{Appendix.}

Partially following \cite{KS} we provide the explicit explanation why the functions $k_{t}^{(n)}$ constructed in Section 4  are Lenard positive, i.e. satisfy (\ref{LP}). The main difference with \cite{KS} is that here we deal only with the contact process on the space $\Gamma_0$ of finite configurations.
We show first that for any starting point $\gamma \in \Gamma_0$ the contact process exists as a Markov process in $\Gamma_0$.
The formal generator $L$ can be interpreted as a generator of a Markov process on the space $\Gamma_0$ of finite configurations, i.e. finite subsets $\gamma \subset \mathbb{R}^d$,
\begin{equation}\label{App-1}
L f(\gamma) = 2 |\gamma| \left\{  \int_{\Gamma_0} f(\gamma') Q(\gamma, d\gamma') - f(\gamma) \right\},
\end{equation}
where $|\gamma|$ is the number of points in the configuration $\gamma$.
The transition kernel $Q(\gamma, d\gamma')$ on $\Gamma_0$ takes the form:
\begin{equation}\label{App-2}
Q(\gamma, d\gamma') = \frac{1}{2|\gamma|} \Big\{ \sum\limits_{x \in \gamma} \delta_{\gamma\backslash x}(d\gamma') + \sum\limits_{y \in \gamma} \int\limits_{\mathbb{R}^d} a(x-y) \delta_{\gamma\cup x} (d\gamma') dx  \Big\}.
\end{equation}
An application of the jump Markov processes theory gives us the existence of  a probability space $(\Omega, {\rm P}, \mathcal{F})$ and a Markov process $\gamma (t) \in \Gamma_0, \, t< \tau_{\infty}$ with the generator \eqref{App-1}, where $\tau_{\infty}$ is the life time of the process ($\tau_{\infty}$ is a random time when the number of particles of $\gamma(t)$ becomes infinite). It is easy to see the regularity of the process, i.e. that
\begin{equation}\label{App-3}
{\rm P} \big( \tau_{\infty} = +\infty   \big) = 1.
\end{equation}
Indeed, for the function $|\gamma|:\Gamma_0 \to \mathbb{N}\cup 0$  we get  $L|\gamma| = 0$.
We have the representation
\begin{equation}\label{App-4}
|\gamma(t)| = |\gamma(0)| + \int\limits_0^t L|\gamma(s)| ds + M_t,
\end{equation}
where $M_t$ and so  $|\gamma(t)| $ are local martingales with the localizing sequence of optional times (stopping times) $\tau_N = \inf \{ t: \, |\gamma(t)| \ge N \}$. From the theory of branching processes it is known that
\begin{equation}\label{App-4bis}
{\rm P} \{ \tau_N  \le t\} \le C(t) (1- e^{-t})^N
\end{equation}
(it is sufficient to consider a pure birth process with the birth rate equal to 1, see \cite{GS}, Chapter VII, Sec. 5).
%So $\tau_N \to \infty$ as $N \to \infty$,
Since ${\rm P}(\tau_{\infty}\le t) \le {\rm P}(\tau_N \le t) \; \forall N \in \mathbb{N}$, then \eqref{App-4bis} implies ${\rm P}(\tau_{\infty}< \infty)=0$,
and the regularity \eqref{App-3} holds. Moreover
$$
N^k \Pr  \{ \tau_N  \le t\} \to 0 \qquad \mbox{for any } \; k>0 \; \mbox{ when } \; N \to \infty,
$$
and using Proposition 1.8 from \cite{CKL} we conclude that the local martingale $ |\gamma(t)| $ is a martingale.

We will give now a constructive description of this process.
Namely, for a given configuration $\gamma(0)=\gamma,\, |\gamma|< \infty$,
the contact process started from $\gamma$ has the following structure. For any point $y \in \gamma$ there is the rate of birth $\lambda_b$ and the rate of death $\lambda_d$. In the critical regime $\lambda_b = \lambda_d$. We put here $\lambda_b = \lambda_d =1$. As a result of birth the point $y$ creates a new point   of a configuration in a random position $x \in \mathbb{R}^d$. The distribution density of $x$ equals to $a(x-y)$. The parent $y$ remains to exist after the birth and can produce new offsprings. The total number $|\gamma(t)|$ of points is an integer-valued birth and death process, i.e. the Markov process on integers $n \in \mathbb{N}\cup \{0\} $ with the rates $\lambda_d \cdot n$ for the transition $n \to n-1$ and $\lambda_b \cdot n$ for $n \to n+1$.
This process is defined for all $t>0$ and $\mathbb{E} |\gamma(t)| = |\gamma(0)|$ in the critical case or $\mathbb{E} |\gamma(t)| = e^{(\lambda_b - \lambda_d) t} |\gamma(0)|$ in the general case.
The number of transitions of the process $\gamma(t)$ on any time interval $[0,t]$ is a.s. finite. For simplicity
we assume that $a(x)$ has a compact support. The general case can be considered using approximation arguments from \cite{KKP}, Corollary 4.1. At time t any point of configuration $\gamma(t)$ is a descendant of one of the points of the initial configuration $\gamma(0)$.

Let the initial configuration $\gamma \in \Gamma_0$ be random and it has a probability distribution  $\mu_0 \in {\cal M}^1_{\rm{corr}}(\Gamma_0)$ such that the correlation functions  $k^{(n)}_0$ are bounded for each $n=1,2, \ldots$. We suppose that the random variable $|\gamma|$, which is equal to the number of particles at time $t=0$ has all finite moments. Then from the theory of branching random processes it follows that the same holds for $|\gamma(t)|$ at any $t>0$.
We denote by  $\mu_t$ the distribution of configurations $\gamma(t) \in \Gamma_0$. The contact process on $\Gamma_0$ with the starting configuration $\gamma \in \Gamma_0$ can be interpreted as a probability measure on the space of finite forests, i.e. the finite sets of trees in space-time (trajectories of the process).
%\textcolor{red}{The existence of the moments for the number of particles process $|\gamma(t)|$ at any $t>0$ implies the existence of the %correlation measure $\rho_t$ for any $t>0$.}

%This construction implies the existence of the correlation measure $\rho_t$ associated with the probability measure $\mu_t$.

%Let a measure $\mu_0$ has a set of correlation functions  $k^{(n)}_0$ bounded for each $n$. Then the equation \eqref{59} has a unique %solution $k_{t}^{(n)}$ and this solution belongs to the same class, with bounds depending on $t$, see \cite{KKP}, Prop. 4.4. For any %positive $G \in B_{bs}(\Gamma_0)$ and $f= KG$ we have
%\begin{equation}\label{App0}
%\mu_t (KG) = \mu_0 (e^{tL} KG) = \langle k_0, K^{-1}e^{tL} KG \rangle =  \langle k_0, e^{t \hat L} G \rangle = \langle k_t, G\rangle.
%\end{equation}

\medskip

In addition to the existence of the Markov process $\mu_t$ on $\Gamma_0$ we need to prove the existence and boundedness of correlation functions corresponding to measures $\mu_t$ at any time $t>0$. Also we have to show that the correlation functions satisfy the differential equations (\ref{59}).

From the construction of the contact Markov process $\gamma(t) \in \Gamma_0$ by the same reason as above
(with the same localizing sequence $\tau_N$) applying Proposition 1.8 from \cite{CKL}
we conclude that the random function
\begin{equation}\label{Mar1}
f(\gamma(t)) - \int\limits_0^t \big( Lf \big) (\gamma(s)) \, ds
\end{equation}
is a martingale for $f= KG$, where $G \in B_{bs}(\Gamma_0)$, $L$ is a formal generator given by (\ref{generator}), the mapping $K$ was defined in (\ref{K-transform}) and $f$ is restricted to $\Gamma_0$, i.e.
$$
f(\gamma) = \sum\limits_{\eta \subseteq \gamma} G(\eta), \qquad \gamma, \eta \in \Gamma_0.
$$
Consequently we have
\begin{equation}\label{Mar2}
\mu_t (f) - \int\limits_0^t \mu_s \big( Lf \big) \, ds = \mu_0 (f)
\end{equation}
for $f = KG$, i.e.
\begin{equation}\label{Mar3}
\mu_t (KG) - \int\limits_0^t \mu_s \big( L KG \big) \, ds = \mu_0 (KG).
\end{equation}
%{\it We need martingale in \eqref{Mar1} only for proving \eqref{Mar2}. Let us try to escape \eqref{Mar1} and use other arguments.
%Instead of formula \eqref{Mar1}, can we use the formula
%$$
%T(t) f - f = \int_0^t T(s) Lf ds?
%$$
%Then this implies formulas \eqref{Mar2} and \eqref{Mar3}.
%It seems that we need here only that functions $f = KG$ belong to the domain of the generator \eqref{App-1}. Is it correct? }
As it was mentioned above $|\gamma(t)|$ has all finite moments for arbitrary $t\geq 0$. The latter fact means that the correlation measure $\rho_t$ corresponding to the probability measure $\mu_{t}$ by the formula
\begin{equation}\label{Mar3bis}
\rho_t (G) \ = \ \mu_t (KG), \quad G \in B_{bs}(\Gamma_0)
\end{equation}
exists, and $\rho_t$ is finite.
%\textcolor{blue}{
%Moreover, the measure $\rho_t$ is locally finite, i.e. $\rho_{t}(A)<\infty$ for any bounded $A\in\mathcal{B}(\Gamma_0)$; normalized, i.e. %$\rho_{t}(\{\emptyset\})=1$; and positive definite, i.e.
%$$
%\rho_{t}(G)\geq 0\quad \text{for all}\quad G \in B_{bs}(\Gamma_0)\quad\text{such that}\quad KG\geq 0.
%$$
%We denote such class of measures on $\Gamma_{0}$ by ${\cal M}_{\rm{lf}}^{+}(\Gamma_{0})$. For the class of locally finite and normalized %measures on $\Gamma_{0}$ we use the notation ${\cal M}_{\rm{lf}}(\Gamma_{0})$. }
%\textcolor{blue}{We write $\rho\in{\cal M}_{\rm{corr}}(\Gamma_{0})$,  if $\rho$ is a correlation measure of a unique probability measure %on $\Gamma_{0}$ with finite moments. Clearly, ${\cal M}_{\rm{corr}}(\Gamma_{0})\subset{\cal M}_{\rm{lf}}^{+}(\Gamma_{0})$.}

%\textcolor{blue}{It is important to note that if correlation measure $\rho_{t}$ is absolutely continuous with respect to the %Lebesgue-Poisson measure, the corresponding density $k_{t}$  is a correlation function of $\mu_{t}$.
%}

It follows from (\ref{generator}) and  (\ref{K-transform})  that if  $G \in B_{bs}(\Gamma_0)$ then $K \hat L G = L K G $ is defined and  $\hat L G \in B_{bs}(\Gamma_0)$. This implies
\begin{equation}\label{Mar4}
\mu_t (KG) - \int\limits_0^t \mu_s \big( K \hat L G \big) \, ds = \mu_0 (KG)
\end{equation}
or equivalently,
\begin{equation}\label{Mar5}
\rho_t (G) - \int\limits_0^t \rho_s \big( \hat L G \big) \, ds = \rho_0 (G).
\end{equation}
The equation (\ref{Mar5}) is a weak form of the equation (\ref{59}) in terms of correlation measures.
%\textcolor{blue}
{From (\ref{Mar5}) it follows that $\rho_t(G)$ is a continuous function of $t $ for any $G$ and we can rewrite  (\ref{Mar5}) as a Cauchy problem
$$
\frac{d \rho_t(G)}{dt} = \rho_t(\hat L G) \quad \mbox{with a given } \; \rho_0(G) \; \mbox{ for } \; t=0.
$$
It should be noted that the above reasoning is valid not only for  $G \in B_{bs}(\Gamma_0)$, but also for the functions $G$ of the form $G=(G_1, \ldots ,G_N)$ with some $N >0$, where each $G_n, \, 1\le n \le N$, is a bounded measurable function of space variables, e.g. $G=(G_1), \ G_1 \equiv 1$ corresponds to $KG=|\gamma|$.}

Next we use the Holmgren's principle, see \cite{Shubin}, to prove that the solution of (\ref{Mar5}) coincides with the strong solution of   \eqref{59}. Notice, that if the initial data $k_0^{(n)} \in L^1 ( (\mathbb{R}^d)^n )$, then we obtain that $k_t^{(n)} \in L^1 ( (\mathbb{R}^d)^n)$ for the strong solution of equation (\ref{59}) at any $t>0$. Here we apply the same reasoning as above in Remark 2.3.
Let us consider the adjoint equation $\frac{\partial G}{\partial t}= \hat L G$ where $\hat L$ is defined in \eqref{prop1}. We consider this equation on the space of finite sequences $G_1, \ldots , G_N$ of bounded (in space variables) functions for some $N>0$ (here $G_n=0$ for $n>N $ by definition). An existence of the solution of the adjoint equation follows from the direct calculation. Then the Holmgren's principle says that the uniqueness of the solution of (\ref{Mar5})  in the space of finite measures follows from the existence of the solution in the space of bounded functions for the adjoint equation. Therefore, being unique the weak solution $\rho_t$ has to be absolutely continuous with respect to the Lebesgue-Poisson measure, and the corresponding density $k_{t}$ a.e. coincides with the strong solution of \eqref{59}.
Since the weak solution corresponds to the evolution of measure and thus satisfies the Lenard positivity condition, then the same is valid for the strong solution as well.

Now we proceed to the case when the initial correlation functions  $k_{0}^{(n)}$ correspond to a measure $\mu_0$ on $\Gamma$ (not on $\Gamma_0$).
Let us consider a sequence of expanding balls $B_r \subset \mathbb{R}^d \ (r \to \infty)$ centered at $0$. Then for any set of initial conditions  $k^{(n)}_0$ corresponding to some initial measure $\mu_0 \in {\cal M}^1_{\rm{corr}}(\Gamma)$ we define
$$
k^{(n)}_{0, B_r} (x_1, \ldots, x_n) = k^{(n)}_{0} (x_1, \ldots, x_n) \prod_{i=1}^n \chi_{B_r} (x_i),
$$
where $\chi_B$ is the indicator of $B$. The set of functions $k^{(n)}_{0, B}$ corresponds to a measure $\mu_{0,B}$ on $\Gamma_0$ which is a direct image of the measure $\mu_0$ under the restriction map $\gamma \to \gamma\cap B$.
If the initial conditions  $k_{0}^{(n)}$ are bounded in space variables, then the time evolution  $k^{(n)}_{t, B}$ of the set  $k^{(n)}_{0, B}$ via  \eqref{59} corresponds to the measure $\mu_{t,B}$, that is the time evolution of the measure  $\mu_{0,B}$ on $\Gamma_0$.
Thus the set $k^{(n)}_{t, B}$ is Lenard positive. To prove that the solution  $k_{t}^{(n)}$ of the Cauchy problem (\ref{59}) is Lenard positive it is sufficient to verify that
\begin{equation}\label{App1}
k^{(n)}_{t, B_r} \to k_{t}^{(n)} \quad \mbox{as }\; r \to \infty
\end{equation}
provided that  $k^{(n)}_{t, B}$ are bounded uniformly in $B$.

To prove \eqref{App1} we use the special property of monotonicity of equation \eqref{59}. Namely, if
$$
k^{(n)}_0 (x_1, \ldots, x_n) \le \tilde k^{(n)}_0 (x_1, \ldots, x_n),
$$
then
$$
k^{(n)}_t (x_1, \ldots, x_n) \le \tilde k^{(n)}_t (x_1, \ldots, x_n) \quad \mbox{for any } \; t>0
$$
for the considered bounded solutions of  \eqref{59}.
This property follows from Lemma \ref{3.2} (it is connected with the special feature of the contact process, when the mortality rate does not depend on the population density). Consequently we get
\begin{equation}\label{App2}
k^{(n)}_{t,B_1} (x_1, \ldots, x_n) \le  k^{(n)}_{t, B_2} (x_1, \ldots, x_n) \le k^{(n)}_{t} (x_1, \ldots, x_n)
\end{equation}
for all $B_1 \subset B_2$.
Since $k^{(n)}_{0, B}$ are bounded for any $n$ and $B$, then $k^{(n)}_{t, B}$ are also bounded for all $B$ and $t \in (0,T)$, see \cite{KKP}, Prop. 4.4. Using monotonicity in $B$ and boundedness \eqref{App2} we conclude that there exists a pointwise limit
$$
\lim_{r\to \infty}k^{(n)}_{t,B_r}  = \bar k^{(n)}_{t}.
$$
The set of functions $\bar k^{(n)}_{t}$ is Lenard positive, and it remains to prove that   $\bar k^{(n)}_{t}= k^{(n)}_{t}$, where $k^{(n)}_{t}$ is a solution of the Cauchy problem \eqref{59} with initial data $k^{(n)}_{0}$. If we rewrite \eqref{59} as the integral equation on $[0,T]$ and use the Lebesgue's dominated convergence theorem we conclude that  $\bar k^{(n)}_{t}$ is a solution of \eqref{59}. The uniqueness of the solution in $\BB((\R^{d})^n)$ implies that  $\bar k^{(n)}_{t}= k^{(n)}_{t}$.

\bigskip

\noindent {\bf Acknowledgement.} The authors are grateful to Prof. Stanislav Molchanov for fruitful discussions during his visits to the University of Bielefeld.\\
Oleksandr Kutovyi gratefully acknowledges the financial support of the German Research Foundation (DFG) through the IRTG 2235 and CRC 1283.
Sergey Pirogov gratefully acknowledges the financial support of the Russian Science Foundation, Project 17-11-01098.\\
Elena Zhizhina expresses her gratitude to the Mathematical Department of Bielefeld University for kind hospitality.


\begin{thebibliography}{20}


%\bibitem{FKK} D. Finkelshtein, Yu. Kondratiev, O. Kutoviy, Semigroup
%approach to non-equilibrium birth-and-death stochastic dynamics in
%continuum, J. Funct. Anal. 262, 1274-1308 (2012)

%\bibitem{KK} Yu. G. Kondratiev, T. Kuna, Harmonic analysis on configuration
%space: I. General theory, Ininite Dimensional Analysis, Quantum
%Probability and Related Topics, Vol. 5, 201-233 (2002)

%\bibitem{BKKL} Yu. Berezansky, Yu. Kondratiev, T. Kuna, E. Lytvynov,
%On a spectral representation for correlation measures in configuration space analysis,
%Methods Funct. Anal. Topology 5 (1999), no. 4, 87–100.

\bibitem{BS} M. Birkner, R. Sun, Low-dimensional lonely branching random walks die out, Ann. Probab.  \textbf{47} (2), p. 774-803, 2019.

\bibitem{CKL} K.L. Chung, R.J. Williams, Introduction to Stochastic Integration, Birkhauser, Boston, 1983.

\bibitem{DS} M. Duneau, B. Souillard, Cluster properties of lattice and continuous systems, Comm. Math. Phys., vol.47,
pp.155-166 (1976).

\bibitem{EN} {K.-J. Engel, R. Nagel}, \emph{One-parameter Semigroups
For Linear Evolution Equations}, Springer, 2000.


\bibitem{GS} I.I. Gikhman, A.V. Skorokhod, Introduction to the Theory of Random Processes, W.B. Saunders Company, 1969.

\bibitem{GW} Gorostiza, L.G., Wakolbinger, A., Persistence criteria for a class of critical branching particle systems in continuous time,
Ann. Probab. \textbf{19} (1991), p. 266-288.


\bibitem{KK2002}
Kondratiev, Y., Kuna, T.: Harmonic analysis on configuration space.
{I}.
  {G}eneral theory.
\newblock Infin. Dimens. Anal. Quantum Probab. Relat. Top. \textbf{5}(2),
  201--233 (2002)

  \bibitem{KK2006}
Y.~Kondratiev and O.~Kutoviy.
\newblock On the metrical properties of the configuration space.
\newblock { Math. Nachr.}, 279\penalty0 (7):\penalty0 774--783, 2006.

\bibitem{KM2008}
Y.~Kondratiev, O.~Kutoviy, and R.~Minlos,
On non-equilibrium stochastic dynamics for interacting particle systems in continuum,
{J.~Funct. Anal.} \textbf{255} (2008), 200--227.

\bibitem{KKP} Yu. Kondratiev, O. Kutoviy, S. Pirogov, Correlation functions
and invariant measures in continuous contact model, Ininite
Dimensional Analysis, Quantum Probability and Related Topics Vol.
11, No. 2, 231-258 (2008)

\bibitem{KKS} Yu. G. Kondratiev, O. V. Kutoviy, S. Struckmeier, Contact model with Kawasaki dynamics in continuum,
SFB-701 Preprint, University of Bielefeld, Bielefeld, Germany (2007).

\bibitem{KMPZ} Yu. Kondratiev, S. Molchanov, S. Pirogov, E. Zhizhina, On ground state of some non local Schrodinger operator,
Applicable Analysis, 96 (8), 2017,  pp. 1390-1400, doi.org/10.1080/00036811.2016.1192138.

\bibitem{KPZh} Yu. Kondratiev, S. Pirogov, E. Zhizhina, A Quasispecies Continuous Contact Model
in a Critical Regime, Journal of Statistical Physics, 163(2), 357-373 (2016), doi:10.1007/s10955-016-1480-5


\bibitem{KS} Yu. G. Kondratiev and A. Skorokhod, On contact processes in continuum,
Ininite Dimensional Analysis, Quantum Probability and Related Topics
Vol. 9, 187-198 (2006)

\bibitem{SK} Koralov L.B., Sinai Y.G., Theory of Probability and Random Processes, Springer, 2007.

%\bibitem{KR} Krein, M.G.; Rutman, M.A., Linear operators leaving
%invariant a cone in a Banach space, Uspehi Matem. Nauk (in Russian)
%3, p. 3-95 (1948). English translation: Krein, M.G.; Rutman, M.A.,
%Linear operators leaving invariant a cone in a Banach space, Amer.
%Math. Soc. Translation 26 (1950).


\bibitem{L1} A. Lenard, Correlation functions and the uniqueness of the state in classical statistical mechanics,
Comm. Math. Phys. 30, 35-44 (1973).

\bibitem{L2} A. Lenard, States of classical statistical mechanical systems of infinitely many particles, Arch. Rational Mech. Anal.
59, II: 240-256 (1975).

\bibitem{Lig1985}
T.~M. Liggett.
\newblock {\em Interacting particle systems}, volume 276 of {\em Grundlehren
  der Mathematischen Wissenschaften [Fundamental Principles of Mathematical
  Sciences]}.
\newblock Springer-Verlag, New York, 1985.
%\newblock ISBN 0-387-96069-4.
%\newblock xv+488 pp.

\bibitem{Shubin} M. Shubin, Invitation to Partial Differential Equations, AMS EPUB online, 2012.

\bibitem{R} D. Ruelle, Statistical Mechanics, Benjamin, 1969.



\end{thebibliography}
\end{document}